\documentclass[journal]{IEEEtran}
\usepackage{array} 
\usepackage[dvips]{graphicx,color}
\usepackage{amsmath}
\usepackage{xcolor,graphicx}
\usepackage{amsthm}
\usepackage{amssymb}
\usepackage{array}
\usepackage{mdwmath}
\usepackage{mdwtab}
\usepackage{amsfonts}
\usepackage{subfigure}
\usepackage{cite}

\usepackage{amsthm}
\usepackage{caption}
\newtheorem{theorem}{\bf{Theorem}}
\newtheorem{assump}{\bf{Assumption}}
\newtheorem{lemma}{\bf{Lemma}}
\newtheorem{remark}{\bf{Remark}}

\newlength{\halfpagewidth}
\divide\halfpagewidth by 2

\def\BibTeX{{\rm B\kern-.05em{\sc i\kern-.025em b}\kern-.08em
T\kern-.1667em\lower.7ex\hbox{E}\kern-.125emX}}
\usepackage{color}

\begin{document}

\renewcommand{\figurename}{Fig.}

\title{  Delay-Compensated Distributed PDE Control of  Traffic  with Connected/Automated Vehicles}
\author{Jie~Qi,
~\IEEEmembership{Member,~IEEE,}  Shurong Mo,
~Miroslav~Krstic~\IEEEmembership{~Fellow,~IEEE} 
\thanks{Jie Qi is with the College of Information Science and Technology,  Engineering Research Center of Digitized Textile and Fashion Technology Ministry of Education, Donghua University, Shanghai, China, 201620
(e-mail: jieqi@dhu.edu.cn). }
\thanks{Shurong Mo is  with the College of Information Science and Technology, Donghua University, Shanghai, China, 201620 (e-mail: shurong\_mo@mail.dhu.edu.cn).}
\thanks{M. Krstic is with
the Department of Mechanical Aerospace Engineering,
University of California, San Diego, CA 92093-0411, USA 
(e-mail: krstic@ucsd.edu)}            
\thanks{This work was partly supported by the National Natural Science Foundation
of China (62173084, 61773112) and State Key Laboratory of Synthetical Automation for Process Industries. }
}

\maketitle

\begin{abstract}
We develop an input  delay-compensating design for stabilization of   an  Aw-Rascle-Zhang (ARZ) traffic model in congested regime  which is governed by  a  $2\times 2$  first-order hyperbolic nonlinear PDE. The traffic flow consists  of both adaptive cruise control-equipped (ACC-equipped) and manually-driven vehicles.                  
The
control input is the  time gap  of ACC-equipped
and connected vehicles, which is subject to  delays  resulting from communication lag.   For the linearized system, a  novel three-branch bakcstepping transformation with explicit kernel functions  is introduced  to compensate the input delay. 
The transformation is proved  to be bounded, continuous and invertible,  with explicit inverse transformation derived.   Based on the transformation, we obtain the explicit predictor-feedback controller.   
We prove exponential stability of the closed-loop  system with the delay compensator in $L_2$ norm. 
The performance improvement of the closed-loop system under the proposed controller is illustrated in simulation. 

\end{abstract}
\begin{IEEEkeywords}
Delayed distributed input, PDE backstepping, Traffic flow, Predictor-feedback, First-order hyperbolic PDE, Adaptive cruise control (ACC)
\end{IEEEkeywords}

\section{Introduction}\label{intrduce}
\IEEEPARstart{T}{raffic} congestion  has become a severe worldwide social issue. Stop-and-go traffic is  a common phenomenon in   congested highway traffic\cite{treiber2014traffic}, which results from a small perturbation,  such as a delay in  a driver's response,  propagating backward in traffic flow\cite{chen2001causes,ngoduy2013instability}. The stop-and-go  oscillation in  traffic flow leads to poorer driving experience, higher fuel consumption and a high accident risk.
One promising way to reduce the  oscillation in the congested regime is to  develop control design tools that exploit the capabilities of automated and connected vehicles, such as  manipulation of the time gap setting of  ACC-equipped and connected vehicles \cite{diakaki2015overview,9018188}.

PDE-based   models have established a realistic description of the traffic dynamics \cite{delis2015macroscopic,delle2017traffic,aw2000resurrection,kolb2017capacity,stern2018dissipation}
by capturing the temporal and spatial dynamics of the traffic density and the traffic speed along the considered highway stretch. 
Boundary control \cite{karafyllis2018feedback,yu2018traffic,yu2019traffic,zhang2017necessary,zhang2019pi,yu2018stabilization}, and in-domain manipulation \cite{darbha1999intelligent}, \cite{yi2006macroscopic} are both developed to stabilize  the  traffic flow.  Traffic state estimations are also considered in different situations  \cite{yu2019traffic}, \cite{bekiaris2017highway}, \cite{fountoulakis2017highway}. Due to the information transmission or (and)  reaction time of drivers \cite{burger2019derivation}, there are usually time delays in the traffic flow control process, see, e.g.,  \cite{jin2014dynamics} for delays in feedback.  Therefore, it is necessary to study the traffic flow control involving delays. However, few papers on   traffic control consider  delay compensation. 

We consider an  Aw-Rascle-Zhang (ARZ) traffic model in congested regime comprising   manually-driven vehicles and  ACC vehicles,  which are subject to input delays. The model is composed of  $2\times 2$ first-order hyperbolic PDEs whose states are traffic density and  traffic speed. In order to eliminate the stop-and-go waves, a  distributed actuation using the  time-gap of the  ACC vehicles is employed for control. If  the actuator delay is present and sufficiently large, the state feedback control proposed in \cite{9018188} may become destabilizing.

Relevant advances have been achieved towards the controlling of diffusion-driven distributed parameter systems with delays. 
Early predictor-based boundary controller is  developed using the PDE backstepping method in  \cite{krstic2009control}, which stabilizes  an unstable reaction-diffusion PDE with arbitrarily long input delay. 
The aforementioned method has been extended to a 3-D formation control problem to compensate for the effect of potential input delays \cite{qi2019control}. 
For state delays in  an unstable reaction-diffusion PDE, a boundary feedback has been developed using  backstepping in \cite{hashimoto2016stabilization}.    
Recently, a delay-compensator  was designed for   an unstable reaction-diffusion PDE via distributed actuation in \cite{qi2019stabilization}. 
Alternatively, series control design approaches based on Lyapunov-Krasovskii functions have been proposed in   \cite{selivanov2018delayed} and \cite{selivanov2019delayed}, respectively. 
A boundary feedback to compensate a constant input delay  for an unstable reaction-diffusion PDE has been developed in \cite{prieur2018feedback}, using spectral reduced-order models which approximate the infinite-dimensional system by a finite-dimensional one.
A similar method is used for in-domain stabilization of reaction-diffusion PDEs  with time-and spatially-varying delays  in \cite{lhachemi2021robustness}.   

Most studies on  delay-compensator design are for parabolic PDEs. 
There are fewer results on hyperbolic PDEs with delays. An example of   first-order hyperbolic PDE is given in \cite{sano2019boundary}, where a backstepping boundary control  is designed for compensating the input delay. State delay and measurement delay is addressed in \cite{qi2021output}.
In the context of robustness analysis, a delay-robust boundary feedback has been proposed for a  $2\times 2$ linear hyperbolic PDEs  in \cite{auriol2018delay}. 
Both in \cite{bekiaris2018compensation} and  \cite{yu2020bilateral}, the authors  introduce  an equivalent delay system representation to first-order hyperbolic PDEs, which transforms the coupled PDE-ODE systems to ODE systems with input delays. 

The overall challenge addressed in this work is the design of an in-domain  delay-compensator for a $2\times 2$  hyperbolic PDE by employing the  PDE backstepping method while dealing with a dynamic boundary condition which results in a $2\times 2$  hyperbolic PDE-ODE cascade system. The usual Volterra integral transformation cannot be applied directly to  control design because the resulting kernel equation is unsolvable. Therefore, we  propose a three-branch affine  Volterra transformation  which contains  the state of the ODE, namely, the traffic speed at the outlet boundary. The transformation with explicit kernel functions has a different form in each  of the three intervals. Although there are three intervals, the transformation is proved to be continuous in  its domain. Further, we derive the explicit inverse transformation which is   bounded and continuous too.       Based on the transformation, we obtain an explicit delay-compensator,  composed of the feedback of the states and the historical actuator state.  The compensator  stabilizes the traffic flow with input delay via manipulation of the time gap of the ACC-equipped vehicles in-domain.         
We prove the closed-loop system is   $L_2$ exponentially stable, by establishing the $L_2$ stability of the target system and the norm equivalence between the target system and the original system based on the fact that the transformation is invertible.

The structure of the paper is as follows. 
In Section \ref{control}, we introduce the model. In Section \ref{section-control}, we design the delay-compensator using the backstepping method.   Section \ref{stable} presents the proof of the $L_2$ norm exponential stability of the close-loop system. 
The effectiveness of the proposed delay-compensated controller is illustrated with numerical simulations in Section \ref{simulation}. 
The conclusion is given in Section \ref{conclusion}.

\textbf{Notation:} Throughout the paper, we adopt the following notation to define $L_2$-norm for $f(\cdot)\in L_2(0,L)$, $g(\cdot,\cdot)\in L_2((0,L)\times(0,D))$ :
\begin{align*}
        &\rVert f\rVert^2_{L_2}=\int_{0}^{L}|f(x)|^2\mathrm{d}x, ~~ \rVert g\rVert^2_{L_2}=\int_{0}^{L}\int_0^D|g(x,s)|^2\mathrm{d}s\mathrm{d}x.
\end{align*}

\section{Model Description}\label{control}
\subsection{ARZ Traffic model with mixed vehicles}
\label{control-general}

We consider the  ARZ traffic model of highway  introduced in  \cite{9018188} but an input delay which acts on  adaptive cruise control-equipped (ACC-equipped)  vehicles is addressed.  
The state variables of the model  are the traffic density $\breve{\rho}(x,t)$ and the traffic speed $\breve{v}(x,t),$ both defined in domain $(x,t)\in[0, L]\times\mathbb{R}^+ $ where $t$ is time,  $x$ is the spatial variable denoting the position on the concerned highway. Constant $L>0 $ denotes the length of the concerned highway stretch. 
Define $\breve{v}(x,t) \in (0,v_{\rm{f}}] $ with $v_{\rm{f}}$ being free-flow speed.  We consider a mixed traffic, consisting of both  manual and ACC-equipped vehicles with the  percentage of ACC-equipped vehicles with
respect to total vehicles being $\alpha$. Let   $\breve h_{\mathrm{acc}}(x,t)$ denote the  time-gap of the   ACC-equipped vehicle at $x$ from its leading vehicle, which is the control input because a vehicle with ACC   can  automatically adjust its  speed to maintain a desired  distance (or, say, a time-gap) from vehicles ahead.
Due to the lag of information transmission from the control center to each individual ACC vehicle, there often exists input delay. 
Expressed with equations, the traffic flow control system we consider is:   
\begin{align}
    \breve{\rho}_t(x,t)=&-\breve{\rho}_x(x,t)\breve{v}(x,t)-\breve{\rho}(x,t)\breve{v}_x(x,t),\label{traffic_1}\\
    \breve{v}_t(x,t)=&-\breve{\rho}(x,t)\frac{\partial V_{\rm{mix}}(\breve{\rho}(x,t),\breve{h}_{\rm{acc}}(x,t-D))}{\partial \breve{\rho}}\breve{v}_x(x,t)\nonumber\\
    &-\breve{v}(x,t)\breve{v}_x(x,t)\nonumber\\
    &+\frac{V_{\rm{mix}}(\breve{\rho}(x,t),\breve{h}_{\rm{acc}}(x,t-D))-\breve{v}(x,t)}{\tau_{\rm{mix}}(\alpha)},\label{traffic_2}\\
    \breve{\rho}(0,t)    =&q_{\rm{in}}/\breve{v}(0,t),\label{traffic_3}\\
    \breve{v}_t(L,t)=&\frac{V_{\rm{mix}}(\breve{\rho}(L,t),\breve{h}_{\rm{acc}}(L,t-D))-\breve{v}(L,t)}{\tau_{\rm{mix}}(\alpha)},\label{traffic_4}
\end{align}
where $D$ is the delay on the domain-wide actuated time gap input, and 
\begin{align}
\tau_{\rm{mix}}(\alpha) = &\frac{1}{\frac{\alpha}{\tau_{\rm{acc}}}+\frac{1-\alpha}{\tau_{\rm{m}}}}, ~0\leq \alpha \leq 1 
\end{align}  
is  time constant for  a mixture traffic  which depends on both  time constant $\tau_{\rm{acc}}$ of ACC vehicles and time constant  $\tau_{\rm{m}}$ of manual vehicles. 
\begin{figure}[h]
        \centering
        \includegraphics[width=0.6\linewidth]{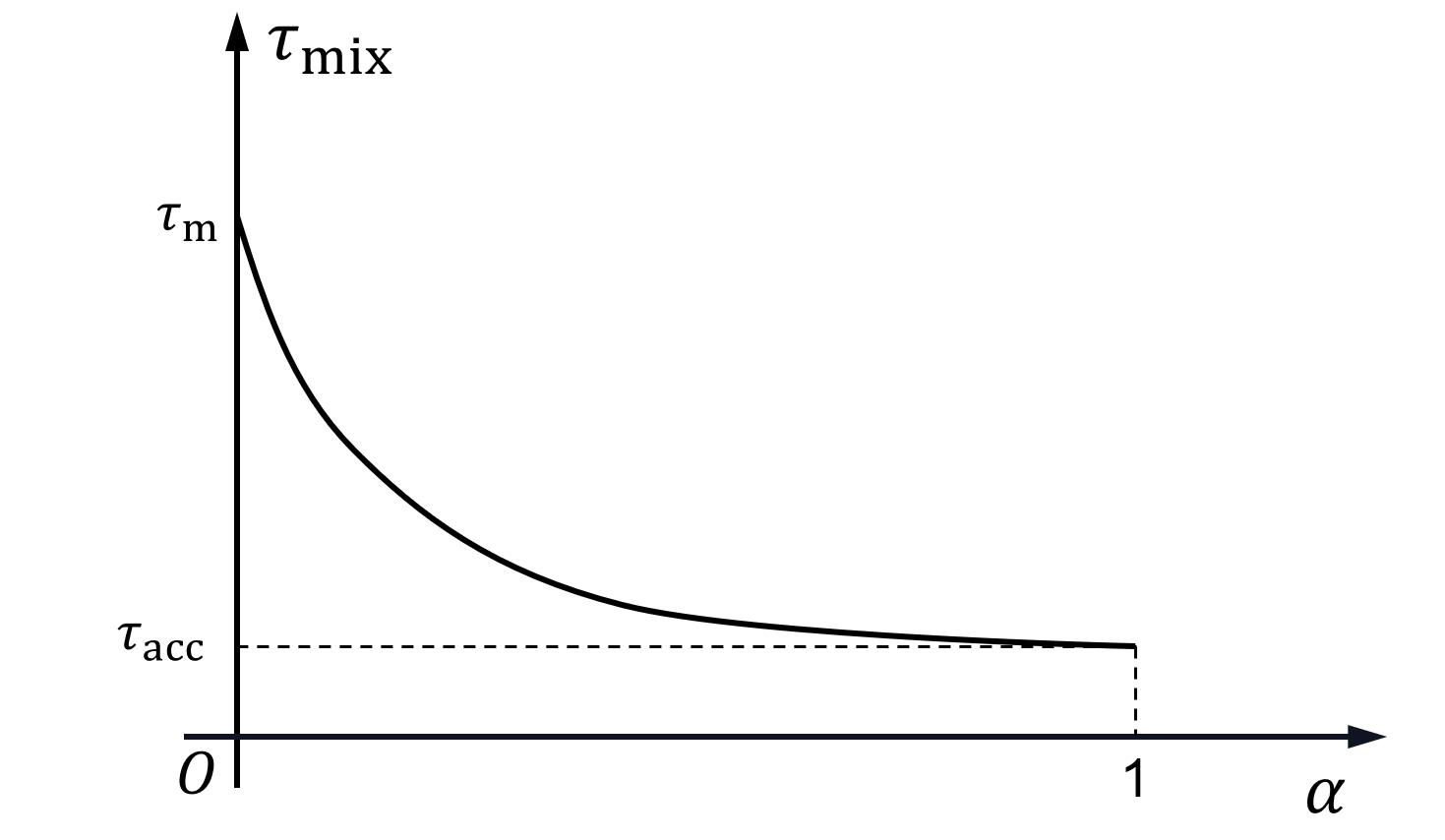}
        \caption{$\tau_{\rm mix}$ varies with $\alpha$.}
        \label{fig:tau_mix}
\end{figure}  
$ \tau_{\rm{mix}}$ is also a function of   $\alpha$, the  percentage of ACC vehicles with
respect to total vehicles. Fig.  \ref{fig:tau_mix} shows the relations of these two parameters.
The equilibrium speed  profile of the mixed flow $V_{\rm{mix}}$ is expressed as 
\begin{align}
V_{\rm{mix}}(\breve{\rho},\breve{h}_{\rm{acc}})=&\frac{1}{\breve{h}_{\rm{mix}}(\breve{h}_{\rm{acc}})}\left (\frac{1}{\breve{\rho}}-l \right ) , \label{v_mix}
\end{align}
and the  mixed time gap is defined as 
\begin{align}
\breve{h}_{\rm{mix}}(\breve{h}_{\rm{acc}})=&\frac{\alpha+(1-\alpha)\frac{\tau_{\rm{acc}}}{\tau_{\rm{m}}}}
{\alpha+(1-\alpha)\frac{\tau_{\rm{acc}}}{\tau_{\rm{m}}}\frac{\breve{h}_{\rm{acc}}}{\breve h_{\rm{m}}}}\breve{h}_{\rm{acc}}.
\label{hmix}
\end{align}
In the above model, $l > 0$ denotes the average effective vehicle length,  $q_{\rm{in}} > 0$ is a constant external inflow,  
and  $\breve h_{\rm{m}} > 0 $ is the time gap of manual vehicles. 

Equation \eqref{traffic_1} means that  the traffic flow observe  the mass conservation law \cite{karafyllis2018feedback}.       Equation \eqref{traffic_2} is a momentum equation inspired by the speed dynamics of ARZ model \cite{zhang2002non} for     ACC-equipped and  manual mixed flow, where $\alpha\in[0,1]$ is the percentage of ACC vehicles with respect to total vehicles. In  \eqref{traffic_2},  
$V_{\mathrm{mix}}(\breve{\rho},\breve{h}_{\rm{acc}})=Q(\breve{\rho},\breve{h}_{\rm{acc}})/\breve{\rho}$ is the equilibrium speed  profile of  a mixed flow of ACC vehicles and manual vehicles, where $Q(\breve \rho, \breve h_{\rm{acc}})$ is the traffic flow given by the fundamental diagram  shown in Fig. \ref{fig:fundamental}.    
Define $\breve \rho_c$ as the lowest density value of  the  mixed time gap $\breve h_{\rm{mix}}$, for which the traffic is congested. 
Let $\breve{h}_{\rm{min}}$ and $\breve{h}_{\rm{max}}$ be the minimum and maximum possible time gap, namely, $ \breve{h}_{\rm{min}} \leq \min\{\breve{h}_{\rm{acc}},\breve{h}_{\rm{m}}\} $ and $ \breve{h}_{\rm{max}} \ge \max\{\breve{h}_{\rm{acc}},\breve{h}_{\rm{m}}\}$.  Define  $\breve{\rho}_{\rm{min}}$ and $\breve{\rho}_{\rm{max}}$ as the lowest  density values of  congested traffic that correspond to minimum and maximum possible time gaps $\breve{h}_{\rm{min}}$ and $\breve{h}_{\rm{max}}$, respectively.  From  Fig. \ref{fig:fundamental}.
we  find  $Q_{\breve h_{\rm{min}}}(\breve \rho_{\rm{min}})=v_{\rm{f}} \breve \rho_{\rm{min}} $ is the maximal flow at  given  time gap $\breve h_{\rm{mix}}\in  [\breve h_{\rm{min}}, 
~\breve h_{\rm{max}}]$. If  $\breve \rho \geq \breve{\rho}_{\rm{min}}$, implying that the traffic is in congested state, we have
\begin{align}
Q_{\breve h_{\rm{mix}}}(\breve \rho)=(1-l\breve{\rho})\frac{v_{\rm{f}}}{1/\breve \rho_{\rm{c}}-l}.
\end{align}
Combined with \eqref{v_mix}, we get
\begin{align}
Q_{\breve h_{\rm{mix}}}(\breve \rho,\breve h_{\rm{acc}})=\breve \rho V_{\rm{mix}}=(1-l\breve{\rho})&\frac{1}{\breve{h}_{\rm{mix}}(\breve{h}_{\rm{acc}})},
\end{align}
which gives
\begin{align}
\breve{h}_{\rm{mix}}(\breve{h}_{\rm{acc}})=\frac{1/\breve \rho_{\rm{c}}-l}{v_{\rm{f}}}.
\end{align}

\begin{figure}[h]
\centering
\includegraphics[width=0.7\linewidth]{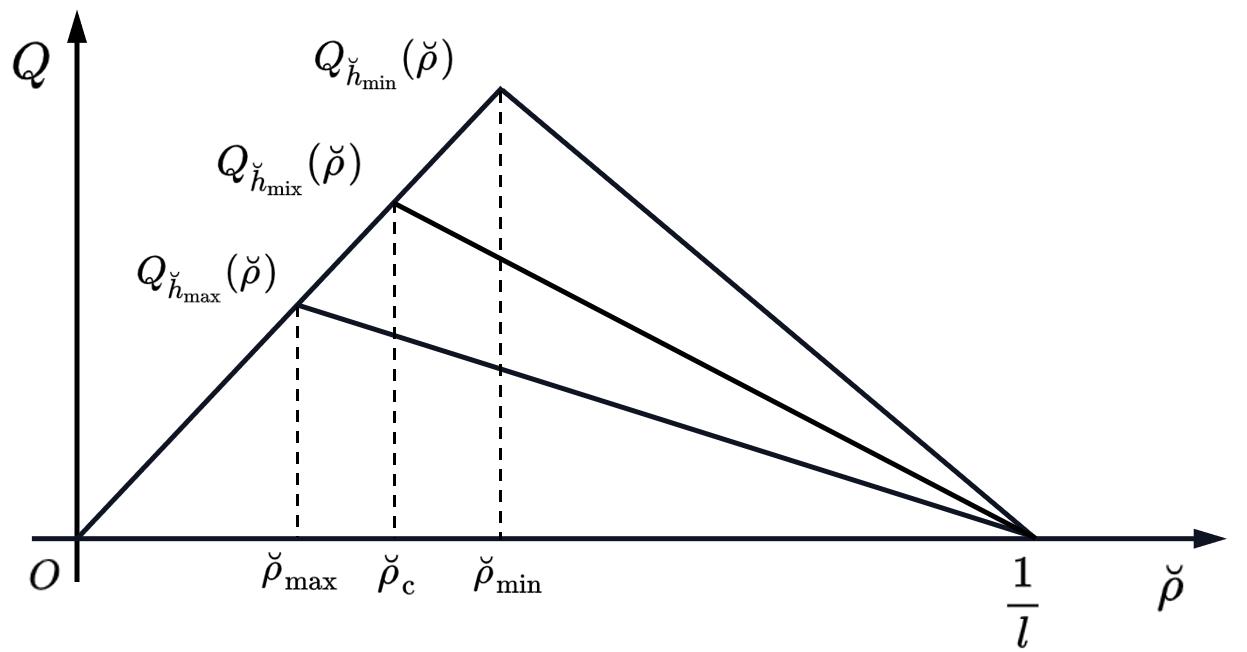}
\caption{Fundamental diagrams for $\breve{h}_{\rm{mix}}\in[\breve{h}_{\rm{min}},\breve{h}_{\rm{max}}]$.}
\label{fig:fundamental}
\end{figure}

In other words,  for each input $\breve h_{\rm{mix}}$, there is a corresponding $\breve \rho_{\rm{c}}$, such as $\breve  \rho_{\rm{min}}=\frac{1}{\breve h_{\rm{min}}v_{\rm{f}}+l}$, which guarantees  $0 < V_{\rm{mix}}(\breve{\rho},\breve{h}_{\rm{acc}}) \leq v_{\rm{f}}$. Fig. \ref{fig:fundamental} shows that possible flow $Q_{\breve h{\rm{mix}}}(\breve \rho)$ for every $\breve h_{\rm{mix}} \in [\breve h_{\rm{min}}, ~\breve h_{\rm{max}}]$ lies between    $Q_{\breve h{\rm{max}}}(\breve \rho)$  and  $Q_{\breve h{\rm{min}}}(\breve \rho)$. In  congested  regime, we define the feasible set of the state and input variables:   $\Phi = \{(\breve{v},\breve{\rho},
~\breve{h}_{\rm{acc}})\in\mathbb{R}^3,~0\leq\breve{v}\leq v_{\rm{f}},
~\breve \rho_{\rm{min}}\leq \breve{\rho}\leq 1/l,
~\breve{h}_{\rm{min}}\le\breve{h}_{\rm{acc}}\le\breve{h}_{\rm{max}}\}$, for all $\alpha\in[0,1]$.  

Using an analysis method similar to the one employed  in \cite{zhang2002non}, one can find that system  \eqref{traffic_1}-\eqref{traffic_4} is anisotropic.

\subsection{Linearization of the ARZ model}        
Consider the  same equilibria of system \eqref{traffic_1}-\eqref{traffic_4} as in \cite{9018188}, dictated by a constant inflow  $ q_{\rm{in}} $ and a constant, steady-state time gap  $ \bar{h}_{\rm{acc}} $ for ACC vehicles, which results in the following  steady-state traffic speed and  density:
\begin{align}\label{steady-state}
\bar{v}=\frac{l}{(1/q_{\rm{in}})-\bar{h}_{\rm{mix}}},~~~\bar{\rho}=\frac{1}{l+\bar h_{\rm{mix}}\bar v}, 
\end{align}
with mixed time gap 
\begin{align}\label{equ-hmix}
\bar{h} _{\rm{mix}}
=&\frac{\alpha+(1-\alpha)
\frac{\tau_{\rm{acc}}}{\tau_{\rm{m}}}}{\alpha+(1-\alpha)
\frac{\tau_{\rm{acc}}}{\tau_{\rm{m}}}\frac{\bar{h}_{\rm{acc}}}
{\breve {h}_{\rm{m}}}}\bar{h}_{\rm{acc}}.
\end{align}
We define the error variables 
\begin{align*}
\rho(x,t)&=\breve{\rho}(x,t)-\bar{\rho}, \\
v(x,t)&=\breve{v}(x,t)-\bar{v},\\  
{h}_{\rm{acc}}(x,t)&=\breve{h}_{\rm{acc}}(x,t)-\bar{h}_{\rm{acc}}.
\end{align*} 
Linearizing (\ref{traffic_1})-(\ref{traffic_4}) around the equilibrium \eqref{steady-state} and \eqref{equ-hmix}, we get  
\begin{align}
\label{equ-zt}
\rho_t(x,t)&=-\bar v\rho_{x}(x,t)-\bar{\rho}v_x(x,t), \\
\label{equ-vt}
v_t(x,t)&=\frac{l}{\bar{h}_{\rm{mix}}}v_{x}(x,t)-\frac{1}{\bar \rho^2 \tau_{\rm{mix}}\bar h_{\rm{mix}}}\rho(x,t)-\frac{1}{\tau_{\rm{mix}}}v(x,t) \nonumber\\
&~~~-\frac{\alpha(1-l\bar\rho)}{\tau_{\rm{acc}}\bar h_{\rm{acc}}^2\bar\rho}h_{\rm{acc}}(x,t-D),\\
\label{equ-bnd-zt}
\rho(0,t)&=-\frac{\bar\rho}{\bar v}v(0,t),\\ 
\label{equ-bnd-vt}
v_t(L,t)&=-\frac{1}{\bar \rho^2 \tau_{\rm{mix}}\bar h_{\rm{mix}}}\rho(L,t)-\frac{1}{\tau_{\rm{mix}}}v(L,t) \nonumber\\
&-\frac{\alpha(1-l\bar\rho)}{\tau_{\rm{acc}}\bar h_{\rm{acc}}^2\bar\rho}h_{\rm{acc}}(L,t-D).
\end{align}      
Introducing a change of variable 
\begin{align}\label{var-change}
{z}(x,t)=\mathrm{e}^{\frac{x}{\bar{v}\tau_{\rm{mix}}}}({\rho}(x,t)+\bar{h}_{\rm{mix}}\bar{\rho}^2{v}(x,t)), \end{align}
and denoting the  input  by  $u(x,t)=h_{\rm{acc}}(x,t)$, 
we obtain a  $2\times2$  first-order hyperbolic linear PDE system in a diagonal form  
\begin{align}
\label{equ-z0}
z_t(x,t)&=-c_1z_{x}(x,t)-\mathrm{e}^{c_2 x}c_3u(x,t-D),\\
\label{equ-v0}
v_t(x,t)&=c_4 v_{x}(x,t)-c_5\mathrm{e}^{-c_2x}z(x,t)-c_6u(x,t-D),\\
\label{equ-bnd-z0}
z(0,t)&= -c_7v(0,t),\\ 
\label{equ-bnd-v0}
v_t(L,t)&=-c_5\mathrm{e}^{-c_2L}z(L,t)-c_6u(L,t-D),
\\\label{equ-initial}
z(x,0)&=z_0(x),
~~v(x,0)=v_0(x),\\\label{equ-ctrlmemory}
u(x,s-D)&=\vartheta_0(x,s),  ~~~~s\in[0,D],
\end{align}
where  
$c_1=\bar v$,
$c_2=\frac{1}{\tau_{\rm{mix}}\bar v}$, 
$c_3=\frac{\alpha\bar h_{\rm{mix}}\bar\rho^2}{\tau_{\rm{acc}}\bar h_{\rm{acc}}^2}((1/\bar \rho)-l)$, 
$c_4=\frac{l}{\bar{h}_{\rm{mix}}}$,
$c_5=\frac{1}{\bar \rho^2 \tau_{\rm{mix}}\bar h_{\rm{mix}}}$,
$c_6=\frac{\alpha}{\tau_{\rm{acc}}\bar h_{\rm{acc}}^2} ((1/\bar \rho)-l)$, 
$c_7=\frac{l\bar \rho^2}{\bar v} $, and one can easily find the equivalence relation of the coefficients: $ c_2c_4=c_5c_7 $ and  $c_1c_2=\frac{c_3c_5}{c_6}$. The   initial conditions is defined in   \eqref{equ-initial}.     The initial actuator state, i.e., the control memory in $[0,D],$ is denoted by $\vartheta_0(x,s)\in L_2([0,L]\times [0,D])$ in \eqref{equ-ctrlmemory}.

Before  we proceed, we  make the following assumption on the coefficients:

\begin{assump}\label{assump-1}
Assume $(c_{1}+c_4)D <L$, which gives $ (c_{1}+c_4)s<L$  for all $ 0\le s \le D $.
\end{assump}
\begin{remark}
\rm{The assumption is  reasonable for the traffic application, because  the length  of the concerned highway stretch $L$ is usually far greater than the other parameters, such that   $(c_{1}+c_4)D$ (delay $D$ times  the sum of  steady speed $c_1=\bar v$  and vehicle length $l$ over mixed time gap $\bar h_{\rm{mix}}$,  $c_4=\frac{l}{\bar{h}_{\rm{mix}}}$) much less than $L$.} 
\end{remark}
Our goal is to find a control $u(x,t)$ that exponentially stabilizes the linearized system \eqref{equ-z0}-\eqref{equ-ctrlmemory} with input delay. In the next section, we present the control design.

\section{Predictor-Feedback Control Design}\label{section-control}
 Before we apply the PDE backstepping approach to the linearized model \eqref{equ-z0}-\eqref{equ-ctrlmemory} with input delay, we first introduce a 2D transport PDE representation of the delay on the 1D-distributed input: 
\begin{align}
\label{equ-z}
z_t(x,t)=&-c_1z_{x}(x,t)-\mathrm{e}^{{c}_2x}c_3\psi(x,0,t),\\
\label{equ-v}
v_t(x,t)=&c_4v_{x}(x,t)-c_{5}\mathrm{e}^{-{c}_2x}z(x,t)-c_6\psi(x,0,t),\\
\label{equ-bnd-z}
z(0 ,t)=&-{c}_7v(0,t),\\ 
\label{equ-bnd-v}
v_{t}(L,t)=&\dot -c_5\mathrm{e}^{-c_2L}z(L,t)-c_6\psi(L,0,t),\\
\label{equ-psi}
\psi_t(x,s,t)=&\psi_s(x,s,t),\\ 
\label{equ-bnd-psi}
\psi(x,D,t)=& u(x,t),\\
\psi(x,s,0)=&\vartheta_0(x,s).\label{equ-initial-psi}
\end{align}
From the last three equations, we have 
\begin{align}\label{equ-r-psi-u}
\psi(x,s,t)=\begin{cases}u(x,t+s-D) & s+t > D \\
\vartheta(x,t+s) & s+t \leq D \\
\end{cases}.
\end{align}
       
\subsection{Backstepping transformation} \label{Backstepping}
To design a stabilizing controller for the PDE-ODE system   \eqref{equ-z}-\eqref{equ-bnd-psi} one has to understand first both its open-loop structure and its actuation structure. Physically speaking, there are three transport processes. Two of the transport processes are in 1D and one is in 2D. One 1D transport is in the $x$ direction at speed $c_1$ in \eqref{equ-z}. The other 1D transport is in the $-x$ direction at speed $c_4$ in \eqref{equ-v}. The two 1D transports create a $(z,v)$ PDE feedback loop, which may be unstable. The 2D transport is in the $-s$ direction at unity speed (and stagnant in the $x$ direction) in \eqref{equ-psi}.

The $z(x,t)$ term in \eqref{equ-v} creates a potentially destabilizing feedback loop but is matched by the actuated term $\psi(x,0,t)$. Likewise, the $z(L,t)$ term in \eqref{equ-bnd-v} creates a potentially destabilizing feedback loop but is matched by the actuated term $\psi(L,0,t)$. These two observations motivate the choice of a target system as 
\begin{align}
    \label{equ-z1}
    {z}_t(x,t)=&-c_1{z}_{x}(x,t)+c_1c_2z(x,t)-c_3\mathrm{e}^{c_2x }
    \beta(x,0,t)\nonumber\\
    &-\frac{kc_3}{c_6}\mathrm{e}^{c_2x }v(x,t),\\
    \label{equ-v1}
    {v}_t(x,t)=&c_4 {v}_{x}(x,t)-c_{6}\beta(x,0,t)-kv(x,t),\\
    \label{equ-bnd-z1}
    {z}(0 ,t)=&-c_7v(0,t),\\ 
    \label{equ-bnd-v1}
    v_t(L,t)=&-kv(L,t)-c_6\beta(L,0,t),\\\label{equ-beta}
    \beta_t(x,s,t)=&\beta_s(x,s,t),\\ 
    \label{equ-bnd-beta}
    \beta(x,D,t)=&0,
\end{align}
where $k>0$  is a free parameter which can be used to  set the desired rate of stability.   

\begin{remark}\label{PDE-ODE}
\rm{The dynamic equation  \eqref{equ-bnd-v} and \eqref{equ-bnd-v1} on the boundary are  not  standard  boundary conditions, which  implies that the  hyperbolic PDE \eqref{equ-z}-\eqref{equ-bnd-psi} and \eqref{equ-z1}--\eqref{equ-bnd-beta} are both preceded by  an ODE whose state is $v(L,t)$. Introducing an additional one-dimensional state $X(t)\in \mathbb{R}$ and  $Y(t)\in \mathbb{R}$  for \eqref{equ-bnd-v} and \eqref{equ-bnd-v1},
respectively,  one can rewrite   \eqref{equ-bnd-v} as:
\begin{align}\label{equ-X}
\dot X=&-c_5\mathrm{e}^{-c_2L}z(L,t)-c_6\psi(L,0,t),
\\
v(L,t)=&X(t).
\end{align}
and \eqref{equ-bnd-v1} as:
\begin{align}\label{equ-Y}
\dot Y=&-kY(L,t)-c_6\beta(L,0,t),
\\
v(L,t)=&Y(t).
\end{align}}
\end{remark} 
Since the additional ODEs are relatively simple, we directly use the boundary values $v(L,t)$ and $v_t(L,t)$ in the following computation for notational brevity.
Introduce the following transformation
\begin{align}
    &\beta(x, s, t)= \psi(x, s, t)+\int_{0}^{L} \gamma(x, s, y) {z}(y, t) d y\nonumber\\
    &+\int_{0}^{L} \eta(x, s, y) {v}(y, t) d y+\mathfrak{r}(x,s){v}(L,t) \nonumber\\
    &+\int_{0}^{L} \int_{0}^{s}G(x,s,y,r) \psi(y, r, t) d r d y,\label{trans-beta1}
\end{align}
where   $\gamma(\cdot,\cdot,\cdot)$, $\eta(\cdot,\cdot,\cdot)$, $G(x,s,y,r)$ and $\mathfrak{r}(\cdot,\cdot)$  are kernel functions defined on $\mathcal{T}_1=\{[0,L]\times [0,D]\times [0,L]\}$,  $\mathcal{T}_{2}=\{(x,s,y,r)|[0, L] \times[0, D] \times[0, L] \times[0, s]\}$ and  $\mathcal{T}_{3}=\{[0, L] \times[0, D]\}$, respectively. 
Mapping   original system \eqref{equ-z}--\eqref{equ-initial-psi} to  target system \eqref{equ-z1}--\eqref{equ-bnd-beta} by transformation \eqref{trans-beta1}, one can get the following equations of $\mathfrak{r}$:
\begin{align}
\mathfrak{r}_s(x,s)=&c_4\eta(x,s,L),\label{equ-r}\\
        \mathfrak{r}(x,0)=&0,
\end{align}
which gives
\begin{align}
\mathfrak{r}(x,s)=&\int_{0}^{s}c_4\eta(x,\theta,L)d\theta,\label{kernel-r}
\end{align}
and equations of $\gamma$ and $\eta$, respectively:
\begin{align}
\gamma_{s}(x,s,y)-c_1\gamma_{y}(x,s,y)
    =&-c_{5}\mathrm{e}^{-{c}_2y}\eta(x,s,y),
    \label{equ-gamma}\\
    \eta_s(x,s,y)+c_4\eta_y(x,s,y)=&~0,\label{equ-eta}\\
    \gamma(x,s,L)=&-\frac{c_{5}}{c_1}\mathrm{e}^{-{c}_2L}\mathfrak{r}(x,s),
    \label{bnd-gamma1}\\
    \gamma(x,0,y)=&-\frac{c_5}{c_6}\mathrm{e}^{-{c}_2x}\delta(x-y),
    \label{bnd-gamma2}\\
    \eta(x,s,0)=&-\frac{c_1c_7}{c_4}\gamma(x,s,0),\label{bnd-eta1}\\
    \eta(x,0,y)=&\frac{k}{c_6}\delta(x-y).
    \label{bnd-eta2}
\end{align}
\begin{figure}[h]
    \centering
    \includegraphics[width=0.8\linewidth]{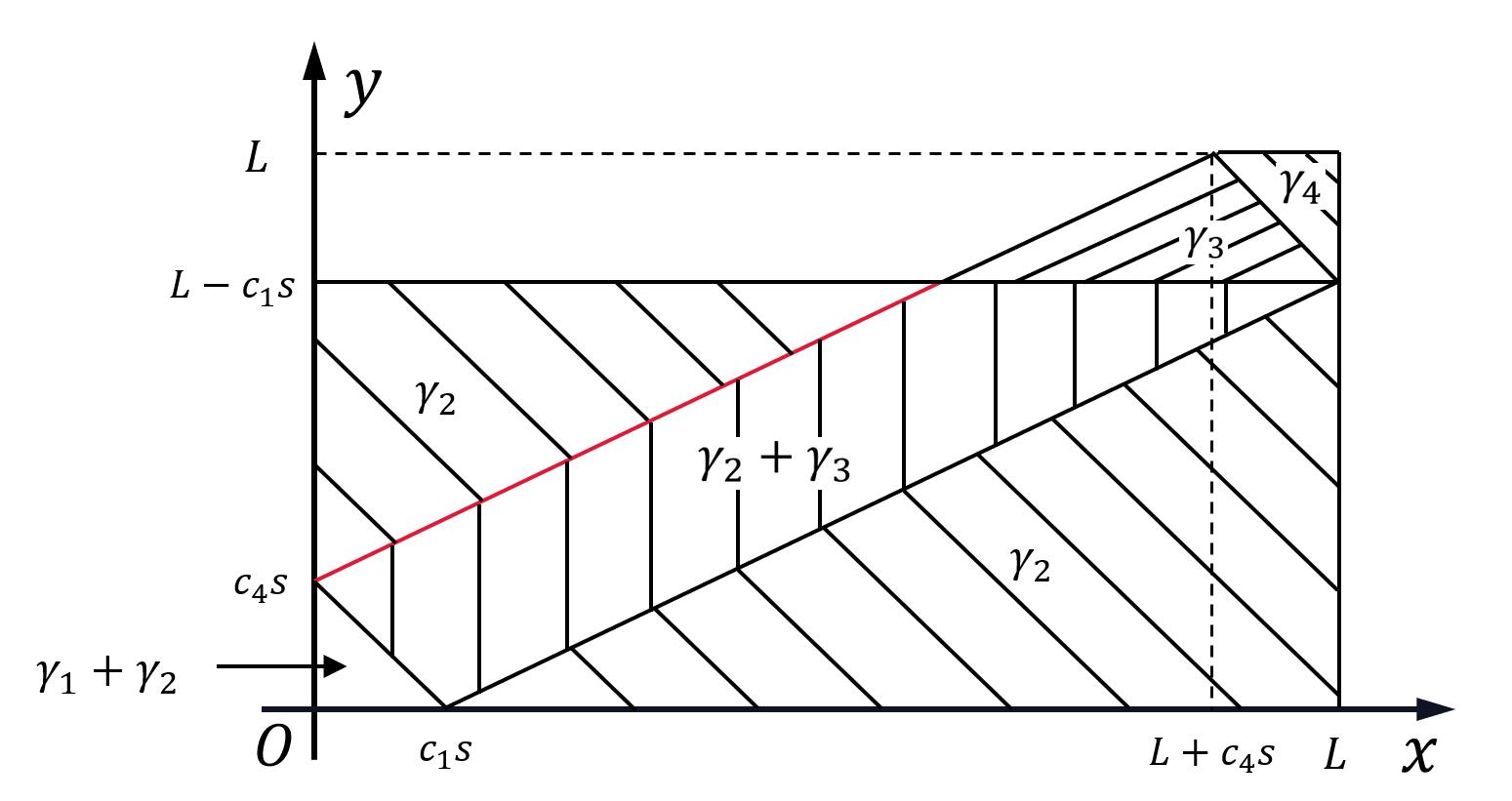}
    \caption{The regions of  kernel $\gamma(x,s,y)$.}
    \label{fig:gamma}
\end{figure} 
Under  Assumption \ref{assump-1}, we solve  kernel $\gamma(\cdot,\cdot,\cdot)$ and $\eta(\cdot,\cdot,\cdot)$  by using the characteristic line method and the successive approximations method (more details please see   Appendix \ref{kernel-solution}) and get: 
\begin{subequations}\label{kernel-gamma}
        \begin{align}
                \gamma&(x,s,y)=
                \gamma_1(x,s,y)+\gamma_2(x,s,y), \\
                &~~\mathrm{if}~0\le y\le c_4(s-\frac{x}{c_1});\nonumber\\
                \gamma&(x,s,y)=
                \gamma_2(x,s,y),\\
                &~~\mathrm{if}~ x+c_4s < y \leq L-c_1s ~\mathrm{or} ~0 \leq y \leq x-c_1s;\nonumber \\
                \gamma&(x,s,y)=
                \gamma_3(x,s,y),\\ 
                &~~\mathrm{if}~ L-c_1s < y \leq \min\{x+c_4s,\frac{c_1+c_4}{c_4}L-\frac{c_1}{c_4}x-c_1s\};\nonumber\\
                \gamma&(x,s,y)=
                \gamma_2(x,s,y)+\gamma_3(x,s,y),\\
                &~~\mathrm{if}~ \max\{c_4(s-\frac{x}{c_1}),x-c_1s\} < y \nonumber\\
                &~~~~~~~~~\leq \min\{x+c_4s,L-c_1s\}; \nonumber\\
                \gamma&(x,s,y)=
                \gamma_4(x,s,y),\\
                &~~\mathrm{if}~ \frac{c_1+c_4}{c_4}L-\frac{c_1}{c_4}x-c_1s\le y\le L; \nonumber\\
                \gamma&(x,s,y)=
                0, ~~\mathrm{otherwise};
        \end{align}
\end{subequations}
with 
\begin{align}
        &\gamma_1(x,s,y)=-\frac{c_5(k+c_5c_7)}{c_6(c_1+c_4)}\mathrm{e}^{-c_2(x+y)},\\
        &\gamma_2(x,s,y)=-\frac{c_5}{c_6}\mathrm{e}^{-{c}_2x}\delta(x-y-c_1s),\\
        &\gamma_3(x,s,y)=-\frac{kc_5}{c_6(c_1+c_4)}\mathrm{e}^{-\frac{c_1c_2}{c_1+c_4}x-\frac{c_2c_4}{c_1+c_4}(y+c_1s)},\\
        &\gamma_4(x,s,y)=-\frac{kc_5}{c_1c_6}\mathrm{e}^{-c_2L}.
\end{align}
\begin{figure}[h]
 \centering
 \subfigure[]{
  \begin{minipage}[ht]{0.6\linewidth}
   \centering
   \includegraphics[width=1.0\textwidth]{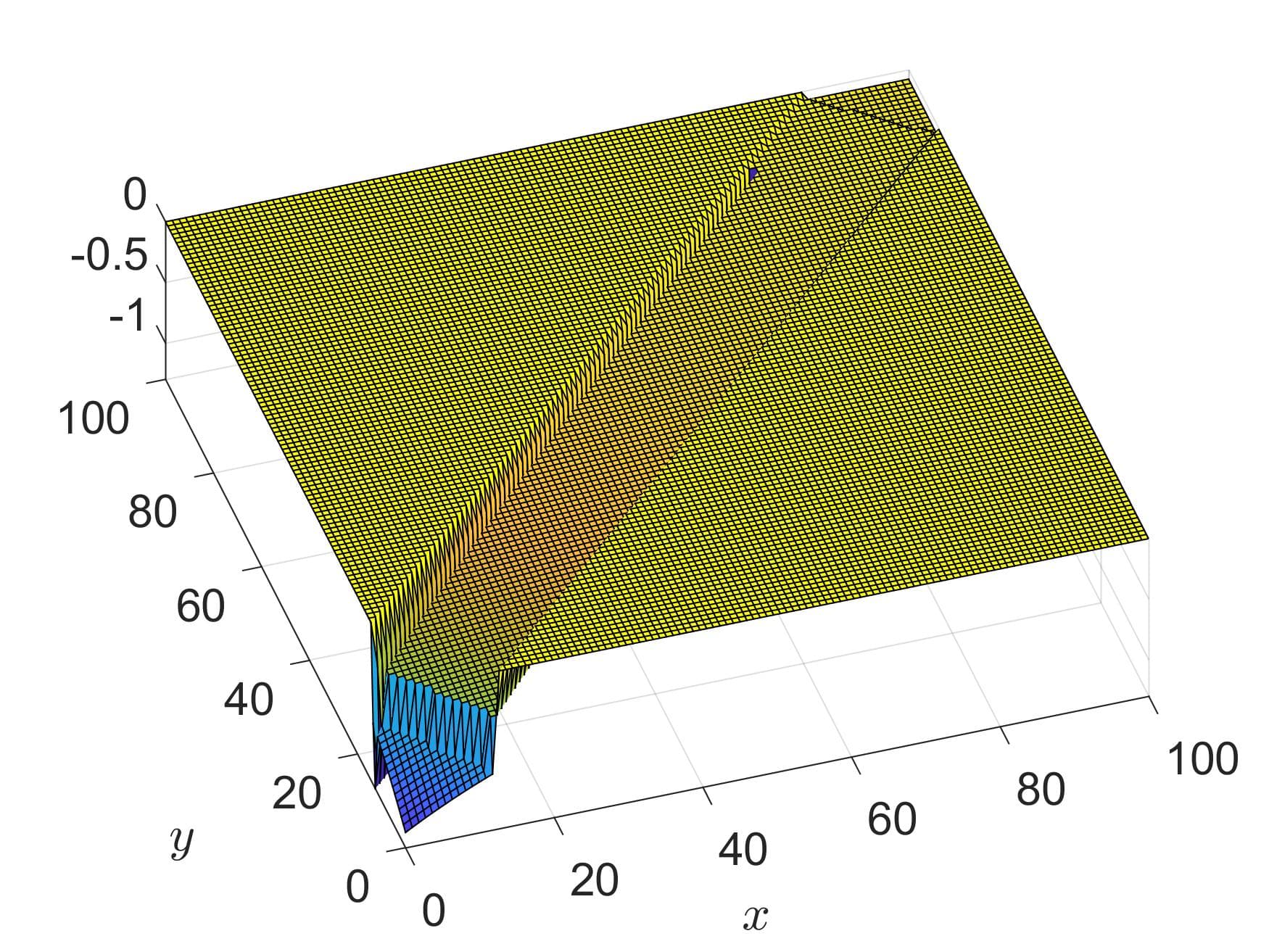}
  \end{minipage}
 }
 \subfigure[]{
  \begin{minipage}[ht]{0.6\linewidth}
   \centering
   \includegraphics[width=1.0\linewidth]{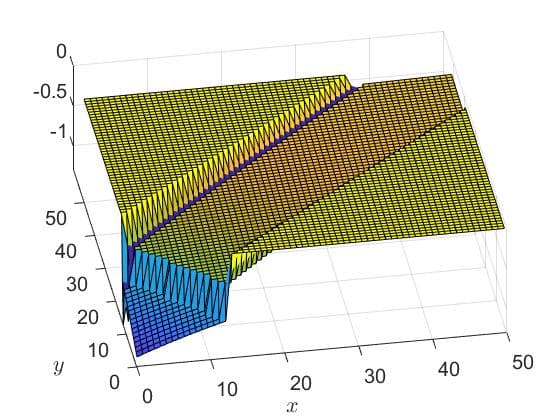}
  \end{minipage}
 }
 \subfigure[]{
  \begin{minipage}[ht]{0.6\linewidth}
   \centering
   \includegraphics[width=1.0\linewidth]{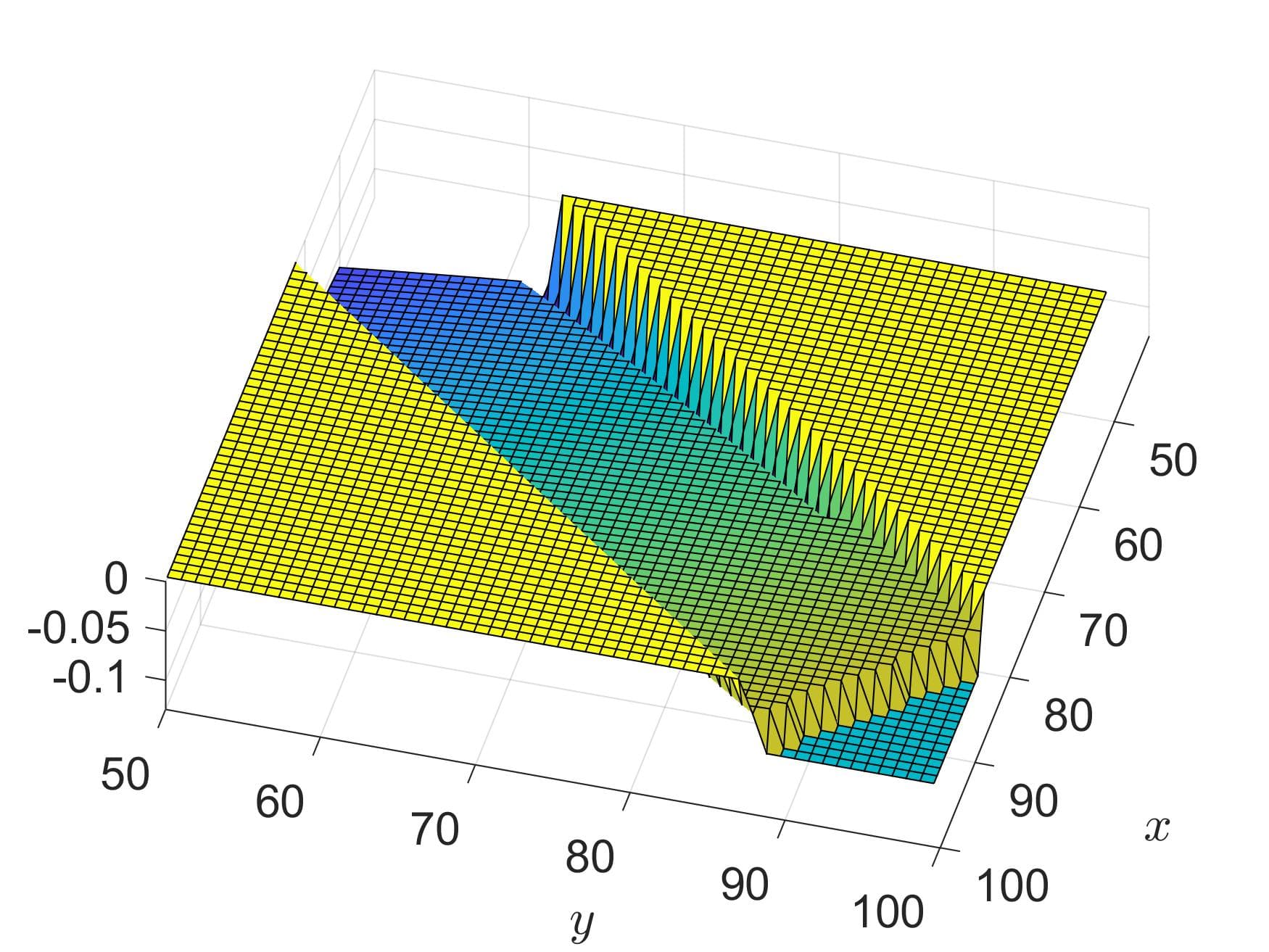}
  \end{minipage}
 } 
 \caption{(a) The kernel function $\gamma(x,s,y)$ with $s=4$ and in domain  $(x,y)\in ([0,100]\times[0,100])$, (b) the partially zoomed (a) in domain $(x,y)\in ([0,50]\times[0,50])$; and (c) the partially zoomed (a) in domain $(x,y)\in ([50,100]\times[50,100])$.}
 \label{fig:1.1}
\end{figure}

Fig. \ref{fig:gamma} shows all the regions that kernel $\gamma(\cdot, s,\cdot)$ takes different value under a given $s$, in which the red line ($y=x-c_1s$) displays where the pulse appears, $\gamma_2\neq 0$.
The  kernel function $\gamma(x,s,y)$ with $s=4$ under  parameters $L=100$, $c_1=3.1048$, $c_2=0.0287$, $c_3=0.0023$, $c_4=3.5981$, 
$c_5=5.5671$, $c_6=0.1438$, $c_7=0.0186$, $k=0.1$, is shown in Fig. \ref{fig:1.1}, where we truncate unlimited pulse of Dirac Delta function  for clear displaying the kernel function. 
Similarly, we get
\begin{subequations}\label{kernel-eta}
        \begin{align}
                \eta&(x,s,y)=
                \eta_1(x,s,y)+\eta_2(x,s,y),\\
                &~~\mathrm{if}~0\le y\le c_4s-\frac{c_4}{c_1}x;\nonumber\\
                \eta&(x,s,y)=\eta_2(x,s,y),\\
                &~~\mathrm{if}~ \max\{c_4s-\frac{c_4}{c_1}x,0\}<y\le c_4s;\nonumber\\
                \eta&(x,s,y)=\eta_3(x,s,y),\\ 
                &~~\mathrm{if}~ c_4s<y\leq L;\nonumber \\
                \eta&(x,s,y)=0, ~\mathrm{otherwise};
        \end{align}
\end{subequations}
with
\begin{align}
    &\eta_1(x,s,y)=\frac{c_1c_5c_7(k+c_5c_7)}{c_4c_6(c_1+c_4)}\mathrm{e}^{-c_2x},\\
    &\eta_2(x,s,y)=\frac{c_1c_5c_7}{c_4c_6}\mathrm{e}^{-c_2x}\delta(x-c_1s+\frac{c_1}{c_4}y),\\
    &\eta_3(x,s,y)=\frac{k}{c_6}\delta(x-y+c_4s).
\end{align}
The equation of $G(\cdot,\cdot,\cdot,\cdot)$ depends on kernel $\gamma(\cdot,\cdot,\cdot)$ and $\eta(\cdot,\cdot,\cdot)$ as follows:

\begin{align}        
        G_s(x,s,y,r)=&-G_r(x,s,y,r)=~0,\label{equ-G}\\
        G(x,s,y,0)=&-c_3\mathrm{e}^{-{c}_2y}\gamma(x,s,y)-c_6\eta(x,s,y)\nonumber\\
        &-c_6\delta(L-y)\mathfrak{r}(x,s),\label{bnd-G}
\end{align}

which is solved 
\begin{align}
                G(x,s,y,r)=&-c_3\mathrm{e}^{-{c}_2y}\gamma(x,s-r,y)-c_6\eta(x,s-r,y)\nonumber\\
        &-c_6\delta(L-y)\mathfrak{r}(x,s-r).\label{kernel-G}       
\end{align}
Substitute the kernel functions \eqref{kernel-gamma}, \eqref{kernel-eta}, \eqref{kernel-r} and \eqref{kernel-G} into the transformation \eqref{trans-beta1}, which gives an explicit backstepping transformation as follows:     
\begin{align}\label{trans-beta}
\beta(x,s,t)=
\begin{cases}
&T_1[\psi(t)](x,s)+Z_1[z(t)](x,s)\\
&~+Y_1[v(t)](x,s),~\mathrm{if}~ 0\le x \le c_1s ,\\
&T_2[\psi(t)](x,s)+Z_2[z(t)](x,s)\\
&~+Y_2[v(t)](x,s), ~\mathrm{if}~ c_1s < x \leq L-c_4s,\\
&T_3[\psi(t)](x,s)+Z_3[z(t)](x,s)\\
&~-\frac{k}{c_6}v(L,t),~\mathrm{if}~ L-c_4s < x \leq L,\\
\end{cases}
\end{align}
where the operators on  state $\psi(\cdot,\cdot, t)$   are
 \begin{align}
T_1&[\psi(t)](x,s)=\psi(x,s,t)\nonumber\\
&-\int_{0}^{\frac{x}{c_1}}c_1c_2\mathrm{e}^{-c_1c_2\tau}\psi(x-c_{1}\tau,s-\tau,t)d\tau\nonumber\\
&+\int_{\frac{x}{c_1}}^sc_5c_7\mathrm{e}^{-c_2x}\psi(c_{4}(\tau-\frac{x}{c_1}),s-\tau,t)d\tau\nonumber\\
&-\int_{0}^s \int_{\max\{x-c_1 \tau,c_{4}(\tau-\frac{x}{c_1})\}}^{x+c_{4}\tau} \mathfrak{g}(x,y,\tau)
\psi(y,s-\tau,t)dyd\tau,\label{T1}\nonumber\\
&+\int_{0}^s k\psi(x+c_4\tau,s-\tau,t)d\tau\\
T_2&[\psi(t)](x,s)=\psi(x,s,t)\nonumber\\
&-\int_0^s c_1c_2\mathrm{e}^{-c_1c_2\tau} \psi(x-c_1\tau,s-\tau,t)d\tau\nonumber\\
&+\int_0^s k \psi(x+c_4 \tau,s-\tau,t)d\tau\nonumber\\
&-\int_0^s \int_{x-c_1 \tau}^{x+c_4 \tau}\mathfrak{g}(x,y,\tau)\psi(y,s-\tau,t)dyd\tau,\label{T2}\\
T_3&[\psi(t)](x,s)=\psi(x,s,t)\nonumber\\
&-\int_0^s c_1c_2\mathrm{e}^{-c_1c_2 \tau} \psi(x-c_1 \tau,s-\tau,t)r\tau\nonumber\\
&+\int_0^{\frac{L-x}{c_4}} k \psi(x+c_4 \tau,s-\tau,t)d\tau\nonumber\\
&+\int_{\frac{L-x}{c_4}}^s k \psi(L,s-\tau,t)d\tau\nonumber\\
&-\int_0^s \int_{x-c_1 \tau}^{\min\{x+c_4 \tau,\mathfrak{c}(\tau)\}}\mathfrak{g}(x,y,\tau) \psi(y,s-\tau,t)dyd\tau\nonumber\\
&-\int_{\frac{L-x}{c_4}}^s \int_{\mathfrak{c}(\tau)}^{L} kc_2\mathrm{e}^{-c_2L} \psi(y,s-\tau,t)dyd\tau,\label{T3}
\end{align}
with 
\begin{align}
\mathfrak{g}(x,y,\tau)&=\frac{kc_1c_2}{c_1+c_4}
\mathrm{e}^{-\frac{c_1c_2}{c_1+c_4}(x-y+c_4\tau)},\\ 
\mathfrak{c}(\tau)&=\frac{c_1+c_4}{c_4}L-\frac{c_1}{c_4}x-c_1 \tau.
\end{align}
The operators on  state $z(\cdot, t)$   are 
\begin{align}
Z_1&[z(t)](x,s)=\int_{0}^{c_4(s-\frac{x}{c_1})}\frac{c_5(k+c_5c_7)}{c_6(c_1+c_4)}\mathrm{e}^{-c_2(x+y)}z(y,t)dy\nonumber \\&+\int_{c_4(s-\frac{x}{c_1})}^{x+c_4s}\mathfrak{k}(x,s)z(y,t)dy,\label{Z1}\\ Z_2&[z(t)](x,s)=\frac{c_5}{c_6}\mathrm{e}^{-c_2x}z(x-c_1s,t)\nonumber\\
&+\int_{x-c_1s}^{x+c_4s}\mathfrak{k}(x,s)z(y,t)dy,\label{Z2}\\
Z_3&[z(t)](x,s)=\frac{c_5}{c_6}\mathrm{e}^{-c_2x}z(x-c_1s,t) \label{Z3}\\
&+\int_{\mathfrak{c}(s)}^{L} \frac{kc_5}{c_1c_6}\mathrm{e}^{-c_2L} z(y,t)dy
+\int_{x-c_1s}^{\mathfrak{c}(s)}\mathfrak{k}(x,s)z(y,t)dy\nonumber
\end{align}
with $\mathfrak{k}(x,s)=\frac{kc_5}{c_6(c_1+c_4)}\mathrm{e}^{-\frac{c_1c_2}{c_1+c_4}x-\frac{c_2c_4(y+c_1s)}{c_1+c_4}}$
. The operators on  state $v(\cdot, t)$  are 
\begin{align}
&Y_1[v(t)](x,s)=-\frac{c_5c_7}{c_6}\mathrm{e}^{-c_2x}v(c_{4}(s-\frac{x}{c_1}) ,t)\nonumber\\
&~~~~-\frac{k}{c_6}v(x+c_4s,t)\nonumber \\
&~~~~-\int_0^{c_4(s-\frac{x}{c_1})}\frac{c_1c_5c_7(k+c_5c_7)}{c_4c_6(c_1+c_4)}\mathrm{e}^{-c_2x} v(y,t)dy,\label{Y1}\\
&Y_2[v(t)](x,s)=-\frac{k}{c_6}v(x+c_4s,t).\label{Y2}
\end{align}        

\begin{lemma}\label{trans-beta-th}
\rm{    If Assumption \ref{assump-1} holds, the transformation \eqref{trans-beta} is bounded and  continuous in $x\in [0,L]$, which 
transforms the original system \eqref{equ-z}-\eqref{equ-bnd-psi} into the  target system \eqref{equ-z1}-\eqref{equ-bnd-beta}.}      
\end{lemma}
The  proof of Lemma  \ref{trans-beta-th} is given in  Appendix \ref{tran-proof}.
\remark
\rm{The transformation of the plant \eqref{equ-z}-\eqref{equ-bnd-psi} to the target system \eqref{equ-z1}-\eqref{equ-bnd-beta} clearly has to be a 2D backstepping transformation due to the 2D nature of the $\psi$-system and the $\beta$-system. 
But it is not just the dimensionality that shall make this 2D backstepping transformation complex. 
The reason for its complexity is that the 1D PDE dynamics of $(z,v)$ evolve perpendicular to the direction of the  2D transport dynamics of $\psi$ through which backstepping is performed. 
The transverse nature of the $(z,v)$-transport relative to the $\beta$-transport, in both the downstream and upstream direction, will make the backstepping transformation $\psi \mapsto \beta$ very complex. }

\subsection{Delay-compensated control}        
The control is obtained by substituting $s=D$ into transformation \eqref{trans-beta}, applying the boundary conditions \eqref{equ-bnd-psi} and \eqref{equ-bnd-beta}, and using the relation \eqref{equ-r-psi-u}, we have if $t>D$,
\begin{subequations}\label{control-u}
        \begin{align}
                u(x,t)=
                &U_1[u](x,t)-Z_1[z(t)](x,D)\nonumber\\
                &-Y_1[v(t)](x,D),\\
                &~~\mathrm{if}~ 0\le x \le c_1D ,\nonumber\\
                u(x,t)=&U_2[u](x,t)-Z_2[z(t)](x,D)\nonumber\\
                &-Y_2[v(t)](x,D),\\
                &~~\mathrm{if}~ c_1D < x \leq L-c_4D,\nonumber \\
                u(x,t)=&U_3[u](x,t)-Z_3[z(t)](x,DL,\tau)d\tau\nonumber\\
                &-\frac{k}{c_6}v(L,t),\\
                &~~\mathrm{if}~ L-c_4D < x \leq L,\nonumber
        \end{align}
\end{subequations}
where 
\begin{align}
&U_1[u](x,t)=\int^{t}_{t-\frac{x}{c_1}}c_1c_2\mathrm{e}^{-c_1c_2(t-\tau)}u(x-c_{1}(t-\tau),\tau)d\tau\nonumber\\
&~~~-\int^{t-\frac{x}{c_1}}_{t-D}c_5c_7\mathrm{e}^{-c_2x}u(c_{4}(t-\tau-\frac{x}{c_1}),\tau)d\tau\nonumber\\
&~~~-\int^{t}_{t-D} ku(x+c_4(t-\tau),\tau)d\tau\nonumber\\
&~~~+\int^{t}_{t-D} \int_{\max\{x-c_1(t-\tau),c_{4}(t-\tau-\frac{x}{c_1})\}}^{x+c_{4}(t-\tau)} \frac{kc_1c_2}{c_1+c_4}\nonumber\\
&~~~\times\mathrm{e}^{-\frac{c_1c_2}{c_1+c_4}(x-y+c_4(t-\tau))} u(y,\tau)dyd\tau,\\
&U_2[u](x,t)=\int_{t-D}^{t} c_1c_2\mathrm{e}^{-c_1c_2(t-\tau)} u(x-c_1(t-\tau),\tau)d\tau\nonumber\\
&~~~-\int_{t-D}^{t} k u(x+c_4(t-\tau),\tau)d\tau\nonumber\\
&~~~+\int_{t-D}^{t} \int_{x-c_1(t-\tau)}^{x+c_4(t-\tau)}\frac{kc_1c_2}{c_1+c_4}\mathrm{e}^{-\frac{c_1c_2}{c_1+c_4}(x-y+c_4(t-\tau))} \nonumber\\
&~~~\times u(y,\tau)dyd\tau,\\
&U_3[u](x,t)=\int_{t-D}^{t} c_1c_2\mathrm{e}^{-c_1c_2(t-\tau)} u(x-c_1(t-\tau),\tau)d\tau\nonumber\\
&~~~-\int^{t}_{t-\frac{L-x}{c_4}} k u(x+c_4(t-\tau),\tau)d\tau\nonumber\\
&~~~+\int_{t-D}^{t} \int_{x-c_1(t-\tau)}^{\min\{x+c_4(t-\tau),\frac{c_1+c_4}{c_4}L-\frac{c_1}{c_4}x-c_1(t-\tau)\}}\frac{kc_1c_2}{c_1+c_4}\nonumber\\
&~~~\times\mathrm{e}^{-\frac{c_1c_2}{c_1+c_4}(x-y+c_4(t-\tau))} u(y,\tau)dyd\tau\nonumber\\
&~~~+\int^{t-\frac{L-x}{c_4}}_{t-D} \int_{\frac{c_1+c_4}{c_4}L-\frac{c_1}{c_4}x-c_1(t-\tau)}^{L} kc_2\mathrm{e}^{-c_2L} u(y,\tau)dyd\tau.
\end{align}
It is because of the aforementioned transverse motion of the $(z,v)$-transport relative to the $\psi$-transport that the control law in \eqref{control-u} is given in three distinct forms:
from the inlet to $c_1D$, from $L-c_4D$ to the outlet, and in between. 
The control \eqref{control-u} is for the linearized and diagonalized  system   \eqref{equ-z}-\eqref{equ-bnd-psi}, and the control law for system \eqref{traffic_1}-\eqref{traffic_4} around the equilibrium \eqref{steady-state} is also required, so we rewrite it as follows: in the case of $t>D$,
\begin{subequations}\label{control-hacc}
        \begin{align}
                \breve h_{\rm{acc}}&(x,t)=
                \bar h_{\rm{acc}}+U_1[\breve h_{\rm{acc}}-\bar h_{\rm{acc}}]-\mathrm{e}^{\frac{c_2}{c_1}x}Z_1[\breve \rho-\bar \rho](x,D)\nonumber\\
                &-\bar h_{\rm{mix}}\bar \rho^2Z_1[\breve v-\bar v](x,D)-Y_1[\breve v-\bar v](x,D),\\
                &~~~\mathrm{if}~ 0\le x \le c_1D ,\nonumber\\
                \breve h_{\rm{acc}}&(x,t)=
                \bar h_{\rm{acc}}+U_2[\breve h_{\rm{acc}}-\bar h_{\rm{acc}}]
                -\mathrm{e}^{\frac{c_2}{c_1}x}Z_2[\breve \rho-\bar \rho](x,D)\nonumber\\
                &~~~-\bar h_{\rm{mix}}\bar \rho^2Z_2[\breve v-\bar v](x,D)
                -Y_2[\breve v-\bar v](x,D),\\
                &~~~\mathrm{if}~ c_1D < x \leq L-c_4D, \nonumber\\
                \breve h_{\rm{acc}}&(x,t)=
                \bar h_{\rm{acc}}+U_3[\breve h_{\rm{acc}}-\bar h_{\rm{acc}}]
                -\mathrm{e}^{\frac{c_2}{c_1}x}Z_3[\breve \rho-\bar \rho](x,D)\nonumber\\
                &~~~-\bar h_{\rm{mix}}\bar \rho^2Z_3[\breve v-\bar v](x,D)
                -\frac{k}{c_6} (\breve v(L,t)-\bar v)\nonumber\\
                &~~~-\int^{t-\frac{L-x}{c_4}}_{t-D} k (\breve h_{\rm{acc}}(L,\tau)-\bar h_{\rm{acc}})d\tau,\\
                &~~~\mathrm{if}~ L-c_4D < x \leq L, \nonumber
        \end{align}       
\end{subequations}
Since the transformation  \eqref{trans-beta}  is continuous in $x$, both the  control \eqref{control-u} for the linearized error system
\eqref{equ-z}-\eqref{equ-bnd-psi} and the control \eqref{control-hacc} for the original system \eqref{traffic_1}-\eqref{traffic_4}  around the equilibrium \eqref{steady-state} are   continuous  in $x$. The control \eqref{control-hacc} is composed of the feedback of the states and the historical actuator state, which  is divided into three parts upon the spatial variable $x$. It implies the ACC vehicles at different position of the highway will apply different time gap  strategies.   Due to the length  of the concerned highway stretch $L$ being  far greater than other parameters (Assumption  \ref{assump-1}), 
one can find that  the first and the last sections are  much shorter than the second section. Hence, most ACC   vehicles  on the highway adopt the second part of control law when they enter the middle  interval $c_1D < x \leq L-c_4D$.   In order to control the traffic flow in the  interval  near  the exit of the highway, the feedback of the flow speed $\breve v(L,t)$ at the exit is also required due to the dynamic boundary condition.

\section{Stability Analysis}\label{stable}
In this section, we  analyze the stability of the closed-loop system. First, we state the main result concerning  exponentially stability.         \begin{theorem}\label{th-hacc}
\rm{  Consider the closed-loop system consisting of plant \eqref{equ-z}
-\eqref{equ-initial-psi} with control law \eqref{control-u}, if the initial conditions $z(\cdot, 0)\in H_1[0,L]$, $v(\cdot, 0) \in H_1[0,L]$, and $\psi(\cdot,\cdot,0))\in  L_2(0,L)\times H_{1}[0,D]$ are compatible, then 
the equilibrium $(z(\cdot, \cdot),v(\cdot, \cdot),\psi(\cdot,\cdot,\cdot))\equiv0$
is exponentially stable in the $L_2$ sense, i.e., there exist  positive constants $\vartheta$ and $M$ such that the following holds for all $t> 0$,                 
\begin{align}
    &\rVert z\rVert^2_{L_2}+\rVert v\rVert^2_{L_2}+| v(L,t)|^2_{}+\rVert \psi\rVert^2_{L_2}+\rVert \psi(L,\cdot,t)\rVert^2_{L_2}\nonumber \\
    \leq& M\mathrm{e}^{-\vartheta t}\left( \rVert z(\cdot, 0)\rVert^2_{L_2}+| v(\cdot, 0)|^2_{}+\rVert v(L, 0)\rVert^2_{L_2}\right.\nonumber \\
    &~\left.+\rVert \psi(\cdot, \cdot, 0)\rVert^2_{L_2}+\rVert \psi(L, \cdot, 0)\rVert^2_{L_2}\right).
\end{align}}
\end{theorem}
The proof of the theorem consists of two steps. The first step is to prove the exponential stability of the target system in $L_2$ sense, and the second step is to  prove the transformation is invertible by obtaining the explicit inverse transformation, which will establish the norm equivalence between the original system and the target system. 
\subsection{The stability of the target system}
Before proceeding, we first define  two Lyapunov functions for the original system \eqref{equ-z}-\eqref{equ-bnd-psi}  and the  target system \eqref{equ-z1}-\eqref{equ-bnd-beta}, respectively
\begin{align}\label{equ-V1}
V_1(t) & = \rVert z\rVert^2_{L_2}+\rVert v\rVert^2_{L_2}+ |v(L,t)|^{2}+\rVert \psi\rVert^2_{L_2}+\rVert \psi(L,\cdot,t)\rVert^2_{L_2},\\
V_2(t)&= \rVert z\rVert^2_{L_2}+\rVert v\rVert^2_{L_2}+ |v(L,t)|^{2}+\rVert \beta\rVert^2_{L_2}+\rVert \beta(L,\cdot,t)\rVert^2_{L_2}\label{equ-V2}.
\end{align}
\begin{lemma}\label{lemma-V0}
\rm{     
        Consider system  (\ref{equ-z1})-(\ref{equ-bnd-beta}). If the initial conditions $ z(\cdot, 0)\in H_1[0,L]$, $v(\cdot, 0) \in H_1[0,L]$, and $\beta(\cdot,\cdot,0))\in  L_2(0,L)\times H_{1}[0,D]$ are compatibles, then the equilibrium  $(z(\cdot, \cdot),v(\cdot, \cdot),\beta(\cdot,\cdot,\cdot))\equiv0$ of  (\ref{equ-z1})-(\ref{equ-bnd-beta}) is exponentially stable in the $L_2$ sense, i.e., there exist  positive constants $\theta$ and $N$ such that the following holds for all $t>0$,
\begin{align}
                &\rVert z\rVert^2_{L_2}+\rVert v\rVert^2_{L_2}+ |v(L,t)|^{2}+\rVert \beta\rVert^2_{L_2}+\rVert \beta(L,\cdot,t)\rVert^2_{L_2}\nonumber \\
                &~~~\leq M\mathrm{e}^{-\theta t}( \rVert z(\cdot, 0)\rVert^2_{L_2}+\rVert v(\cdot, 0)\rVert^2_{L_2}+ |v(L,0)|^{2}\nonumber \\
                &~~~~~~~+\rVert \beta(\cdot, \cdot, 0)\rVert^2_{L_2}+\rVert \beta(L, \cdot, 0)\rVert^2_{L_2}).\label{target-iequ}
\end{align}
}
\end{lemma}
\begin{proof}        
First, define a Lyapunov function for  target system \eqref{equ-z1}-\eqref{equ-bnd-beta}
\begin{align}
V_0(t) =&  \int_{0}^{L}\mathrm e^{-\sigma x}{z}^2(x,t)dx+k_2\int_{0}^{L}\mathrm e^{\sigma x}{v}^2(x,t)dx\nonumber 
\\&+k_3v^2(L,t)+k_4\int_{0}^{L}\int_{0}^{D}\mathrm e^{\sigma (s+x)}\beta ^2(x,s,t)dsdx
\nonumber \\
&+k_5\int_{0}^{D}\mathrm{e}^{\sigma s}\beta ^2(L,s,t)ds,\label{stable-v0}
\end{align}
where  $k_i$, $\sigma>0$, $i=2,3,4,5$, whose ranges  will be determined later.

Differentiating \eqref{stable-v0} with respect to time, we get
\begin{align}
\dot{V_0} (t)={I}(t)+{II}(t)+{III}(t),
\end{align}
where
\begin{align}
    {I}(t)=&-2 \int_{0}^{L}c_1\mathrm e^{-\sigma x}z(x,t){z}_{x}(x,t) dx\nonumber\\
    &2 \int_{0}^{L}c_1c_2\mathrm e^{-\sigma x}z^2(x,t)dx\nonumber\\
    &-2 \int_{0}^{L}c_3\mathrm e^{(c_2-\sigma) x}z(x,t)\beta(x,0,t) dx\nonumber\\
    &-2 \int_{0}^{L}\frac{kc_3}{c_6}\mathrm e^{(c_2-\sigma) x}z(x,t)v(x,t)dx,\\
    {II}(t)=&2k_2\int_{0}^{L}c_4e^{\sigma x}v(x,t)v_{x}(x,t)dx\nonumber\\
    &-2k_2\int_{0}^{L}k\mathrm e^{\sigma x}{v}^2(x,t)dx\nonumber\\
    &-2k_2\int_{0}^{L}c_6\mathrm e^{\sigma x}v(x,t)\beta(x,0,t)dx\nonumber\\
    &-k_3c_6v(L,t)\beta(L,0,t)-2k_3kv^2(L,t),\label{equ-II}\\
    {III}(t)=&-k_4\sigma\int_{0}^{L}\int_{0}^{D}\mathrm e^{\sigma(x+s)}\beta^2(x,s,t)dx\nonumber\\
    &-k_4\int_{0}^{L}\mathrm e^{\sigma x}\beta^2(x,0,t)dx\nonumber \\
    &-k_5\sigma\int_{0}^{D}\mathrm{e}^{\sigma s}\beta^2(L,s,t)ds-k_5\beta^2(L,0,t).\label{equ-III}
\end{align}
By using  Cauchy-Schwarz Inequality, Young's Inequality and letting  $ E=\sup_{x\in[0,L]} \{ \mathrm e^{c_2x} \} $, we get
\begin{align}
&{I}(t)\le
\frac{2Ekc_3}{c_6}\int_{0}^{L}\mathrm e^{\sigma x}v^2dx+ c_1c_7^2v^2(0,t)\nonumber\\
&~~- \left(c_1\sigma-2c_1c_2-\frac{Ec_3}{2}-\frac{Ekc_3}{2c_6}\right)\int_{0}^{L}\mathrm e^{-\sigma x}z^2dx\nonumber\\
&~~+2Ec_3\int_{0}^{L}\mathrm e^{\sigma x}\beta^2(x,0,t)dx,\label{equ-I}\\
&{II}(t)\le
-k_2(c_4\sigma -\frac{c_6}{2}+2k)\int_{0}^{L}\mathrm e^{\sigma x}v^2dx-k_2c_4v^2(0,t)\nonumber\\
&~~+2k_2c_6\int_{0}^{L}\mathrm e^{\sigma x}\beta ^2(x,0,t)dx
+\frac{p}{2}k_3c_6\beta^2(L,0,t)\nonumber\\
&~~-(2k_3k-\frac{k_3c_6}{2p}-k_2c_4\mathrm{e}^{\sigma L})v^2(L,t),
\end{align}
where  $\int_{0}^{L}\mathrm e^{-\sigma x}\beta^2(x,0,t)dx\leq \int_{0}^{L}\mathrm e^{\sigma x}\beta^2(x,0,t)dx$ is used for getting \eqref{equ-I} and $p>0$ is a free parameter of Young's Inequality.
Thus, we get  
\begin{align}
&\dot{V_0}(t)\le
-\left[k_2(c_4\sigma-\frac{c_6}{2}+2k)-\frac{2Ekc_3}{c_6}\right]\int_{0}^{L}\mathrm e^{\sigma x}v^2dx\nonumber\\
&~~- \left(c_1\sigma-2c_1c_2-\frac{Ec_3}{2}-\frac{Ekc_3}{2c_6}\right)\int_{0}^{L}\mathrm e^{-\sigma x}z^2dx\nonumber\\
&~~-k_4\sigma\int_{0}^{L}\int_{0}^{D}\mathrm e^{\sigma(x+s)}\beta^2(x,s,t)dx\nonumber\\
&~~-\left(2k_3k-\frac{k_3c_6}{2p}-k_2c_4\mathrm{e}^{\sigma L}\right)v^2(L,t)\nonumber \\&~~-\left(k_4-2Ec_3-2k_2c_6\right)\int_0^L\beta^2(x,0,t)\nonumber\\
&~~-(k_2c_4- c_1c_7^2)v^2(0,t)-k_5\sigma\int_{0}^{D}\mathrm{e}^{\sigma s}\beta ^2(L,s,t)ds\nonumber\\
&~~-(k_5-\frac{p}{2}k_3c_6)\beta^2(L,0,t).
\end{align}
Choose the  parameters as: 
\begin{align}
        \label{range-parameter}
        \begin{cases}
                \sigma&>2c_2+\frac{Ec_3}{2c_1}+\frac{Ekc_3}{2c_1c_6} , \\
               p &> \frac{c_6}{4k},\\
k_2&>\max\{\frac{2E kc_3}{c_6(c_4\sigma -\frac{c_6}{2}+2k)},
 \frac{ c_1c_7^2}{c_4}\},
 \\ k_3&>\frac{k_2c_4\mathrm{e}^{\sigma L}}{2k-\frac{c_6}{2p}}, \\ k_4&>2E c_3+2k_2c_6,\\
                k_5&>\frac{2k_3c_6}{p},
        \end{cases}
\end{align}such that        
\begin{align*}
&\dot{V_0}(t)\le
-\left[k_2(c_4\sigma -\frac{c_6}{2}+2k)-\frac{2E kc_3}{c_6}\right]\int_{0}^{L}
\mathrm e^{\sigma x}v^2dx\nonumber\\
&~~- \left(c_1\sigma-2c_1c_2-\frac{Ec_3}{2}-\frac{Ekc_3}{2c_6}\right)
\int_{0}^{L}\mathrm e^{-\sigma x}z^2dx\nonumber\\
&~~-k_4\sigma\int_{0}^{L}\int_{0}^{D}\mathrm e^{\sigma(x+s)}
\beta^2(x,s,t)dx\nonumber\\
&~~-\left[k_3(2k-\frac{c_6}{2p})-k_2c_4\mathrm{e}^{\sigma L}\right]
v^2(L,t)\nonumber\\
&~~-k_5\sigma\int_{0}^{D}\mathrm{e}^{\sigma s}\beta ^2(L,s,t)ds\leq -\vartheta V(t) ,
\end{align*} 
where 
\begin{align*}
\vartheta=&\min \{k_2(c_4\sigma -\frac{c_6}{2}+2k)-\frac{2E kc_3}{c_6},\\
& k_4\sigma,~ (c_1\sigma-2c_1c_2-\frac{Ec_3}{2}-\frac{Ekc_3}{2c_6}),~
\\
&k_3(2k-\frac{c_6}{2})-k_2c_4\mathrm{e}^{\sigma L},~
k_5\sigma\}. 
\end{align*}
Therefore, we get
\begin{align}
V_0(t)\le& V_0(0)\mathrm{e}^{-\vartheta t}.
\end{align}      
It is obvious that   $V_0(t)$ defined by \eqref{stable-v0} is equivalent to  $V_2$  defined by \eqref{equ-V2}, i.e., there exist positive constants $\varrho_{1}$ and $\varrho_2$,  such that $\varrho_{1}V_2\leq V_0\leq \varrho_{2} V_2 $. Hence,  \eqref{target-iequ} is proved. 
\end{proof}

\subsection{The inverse transformation }       
\begin{lemma}\label{inverse-lamma}
\rm{  The transformation \eqref{trans-beta} is invertible, and whose inverse transformation is  
        \begin{subequations}\label{inverse-psi}
                \begin{align}
                        \psi(x,s,t)=
                        &Q_1[\beta(t)](x,s)+R_1[v(t)](x,s),\\ 
                        &~~0\le x < c_1s; \nonumber\\
                        \psi(x,s,t)=
                        &Q_2[\beta(t)](x,s)+R_2[v(t)](x,s)\nonumber\\ 
                        &+\mathcal{B}[z(t)](x,s),\\
                        &~~c_1s<x\le L-c_4s;\nonumber\\ 
                        \psi(x,s,t)=
                        &Q_3[\beta(t)](x,s)+R_3[v(t)](x,s)\nonumber\\ 
                        &+\mathcal{B}[z(t)](x,s),\\ 
                        &~~L-c_4s<x\le L; \nonumber 
                \end{align}
        \end{subequations}
where
\begin{align}
&Q_1[\beta(t)](x,s)=\beta(x,s,t)\nonumber\\
&~~~-\int_{\frac{x}{c_1}}^{s}\int_{c_4(\tau-\frac{x}{c_1})}^{x+c_4\tau} \frac{kc_1c_2}{c_1+c_4}\mathrm{e}^{\frac{k(x-y-c_1\tau)}{c_1+c_4}}\beta(y,s-\tau,t)dyd\tau\nonumber\\
&~~~-\int_{0}^{\frac{x}{c_1}}\int_{x-c_1\tau}^{x+c_4\tau} \frac{kc_1c_2}{c_1+c_4}\mathrm{e}^{\frac{k(x-y-c_1\tau)}{c_1+c_4}}\beta(y,s-\tau,t)dyd\tau\nonumber\\
&~~~-\int_{\frac{x}{c_1}}^{s}c_2c_4\mathrm{e}^{k(\frac{x}{c_1}-\tau)}\beta (c_4(\tau-\frac{x}{c_1}),s-\tau,t)d\tau\nonumber\\
&~~~+\int^{\frac{x}{c_1}}_{0}c_1c_2\beta (x-c_1\tau,s-\tau,t)d\tau\nonumber\\
&~~~-\int_{0}^{s}k\mathrm{e}^{-k\tau}\beta (x+c_4\tau,s-\tau,t)d\tau,\\
&Q_2[\beta(t)](x,s)=\beta(x,s,t)\nonumber\\
&~~~-\int_{0}^{s}\int_{x-c_1\tau}^{x+c_4\tau} \frac{kc_1c_2}{c_1+c_4}\mathrm{e}^{\frac{k(x-y-c_1\tau)}{c_1+c_4}}\beta(y,s-\tau,t)dyd\tau\nonumber\\
&~~~+\int_{0}^{s}c_1c_2\beta(x-c_1\tau,s-\tau,t)d\tau\nonumber\\
&~~~-\int_{0}^{s}k\mathrm{e}^{-k\tau}\beta (x+c_4\tau,s-\tau,t)d\tau,\\
&Q_3[\beta(t)](x,s)=\beta(x,s,t)\nonumber\\
&~~~-\int_{0}^{s}\int_{x-c_1\tau}^{\min\{L,x+c_4\tau\}} \frac{kc_1c_2}{c_1+c_4}\mathrm{e}^{\frac{k(x-y-c_1\tau)}{c_1+c_4}}\nonumber\\
&~~~\times\beta(y,s-\tau,t)dyd\tau\nonumber\\
&~~~+\int_{0}^{s}c_1c_2\beta(x-c_1\tau,s-\tau,t)d\tau\nonumber\\
&~~~-\int_{\frac{L-x}{c_4}}^{s}(c_1c_2\mathrm{e}^{\frac{k(x-L-c_1\tau)}{c_1+c_4}}-c_1c_2\mathrm{e}^{-k\tau}+k\mathrm{e}^{-k\tau})\nonumber\\
&~~~\times\beta (L,s-\tau,t)d\tau\nonumber\\
&~~~-\int_{0}^{\frac{L-x}{c_4}}k\mathrm{e}^{-k\tau}\beta (x+c_4\tau,s-\tau,t)d\tau,\\
&R_1[v(t)](x,s)=\frac{c_2c_4}{c_6}\mathrm{e}^{k(\frac{x}{c_1}-s)}v(c_4(s-\frac{x}{c_1}),t)\nonumber\\
&~~~+\frac{k}{c_6}\mathrm{e}^{-ks}v(x+c_4s,t)\nonumber\\
&~~~+\int_{c_4(s-\frac{x}{c_1})}^{x+c_4s} \frac{kc_1c_2}{c_6(c_1+c_4)}\mathrm{e}^{\frac{k(x-y-c_1s)}{c_1+c_4}}v(y,t)dy,\\
&R_2[v(t)](x,s)=\frac{k}{c_6}\mathrm{e}^{-ks}v(x+c_4s,t)\nonumber\\
&~~~+\int_{x-c_1s}^{x+c_4s} \frac{kc_1c_2}{c_6(c_1+c_4)}\mathrm{e}^{\frac{k(x-y-c_1s)}{c_1+c_4}}v(y,t)dy,\\
&R_3[v(t)](x,s)=\int_{x-c_1s}^{L} \frac{kc_1c_2}{c_6(c_1+c_4)}\mathrm{e}^{\frac{k(x-y-c_1s)}{c_1+c_4}}v(y,t)dy\nonumber\\
&~~~+\left(\frac{c_1c_2}{c_6}\mathrm{(e}^{\frac{k(x-L-c_1s)}{c_1+c_4}}-\mathrm{e}^{-ks})+\frac{k}{c_6}\mathrm{e}^{-ks}\right)v(L,t),\\
&\mathcal{B}[z(t)](x,s)=-\frac{c_5}{c_6}\mathrm{e}^{c_2(c_1s-x)}z(x-c_1s,t).
\end{align} 
The inverse transformation is bounded and  continuous in $x\in [0,L]  $.}   
\end{lemma}
The  proof of Lemma \ref{inverse-lamma} is similar to the proof of Lemma 1, so we will omit the proof due to limited space.

\begin{figure*}[t]
        \centering
        \subfigure[]{
                \begin{minipage}[ht]{0.3\linewidth}
                        \centering
                        \includegraphics[width=1.0\textwidth]{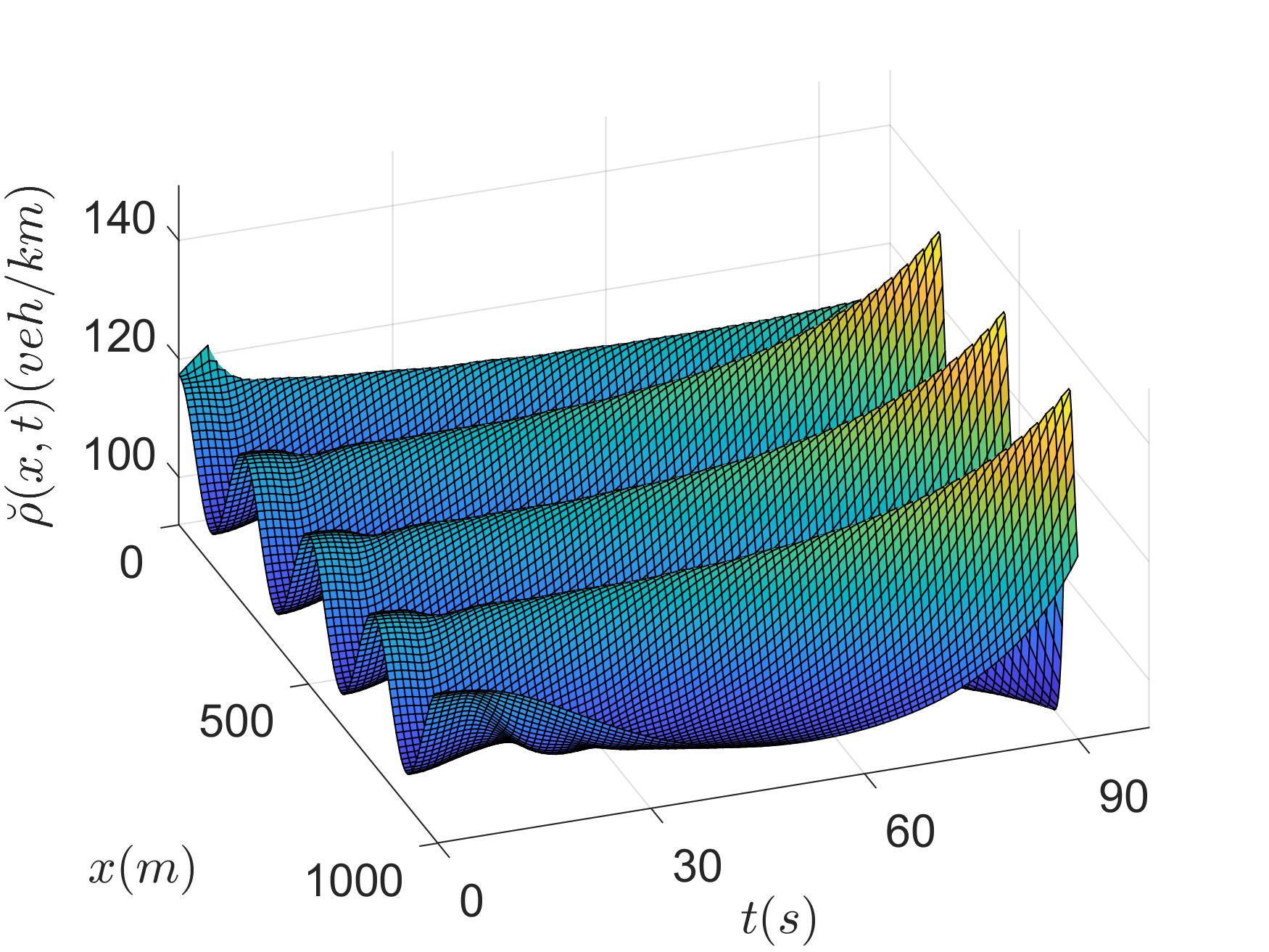}
                \end{minipage}
        }
        \subfigure[]{
                \begin{minipage}[ht]{0.3\linewidth}
                        \centering
                        \includegraphics[width=1.0\linewidth]{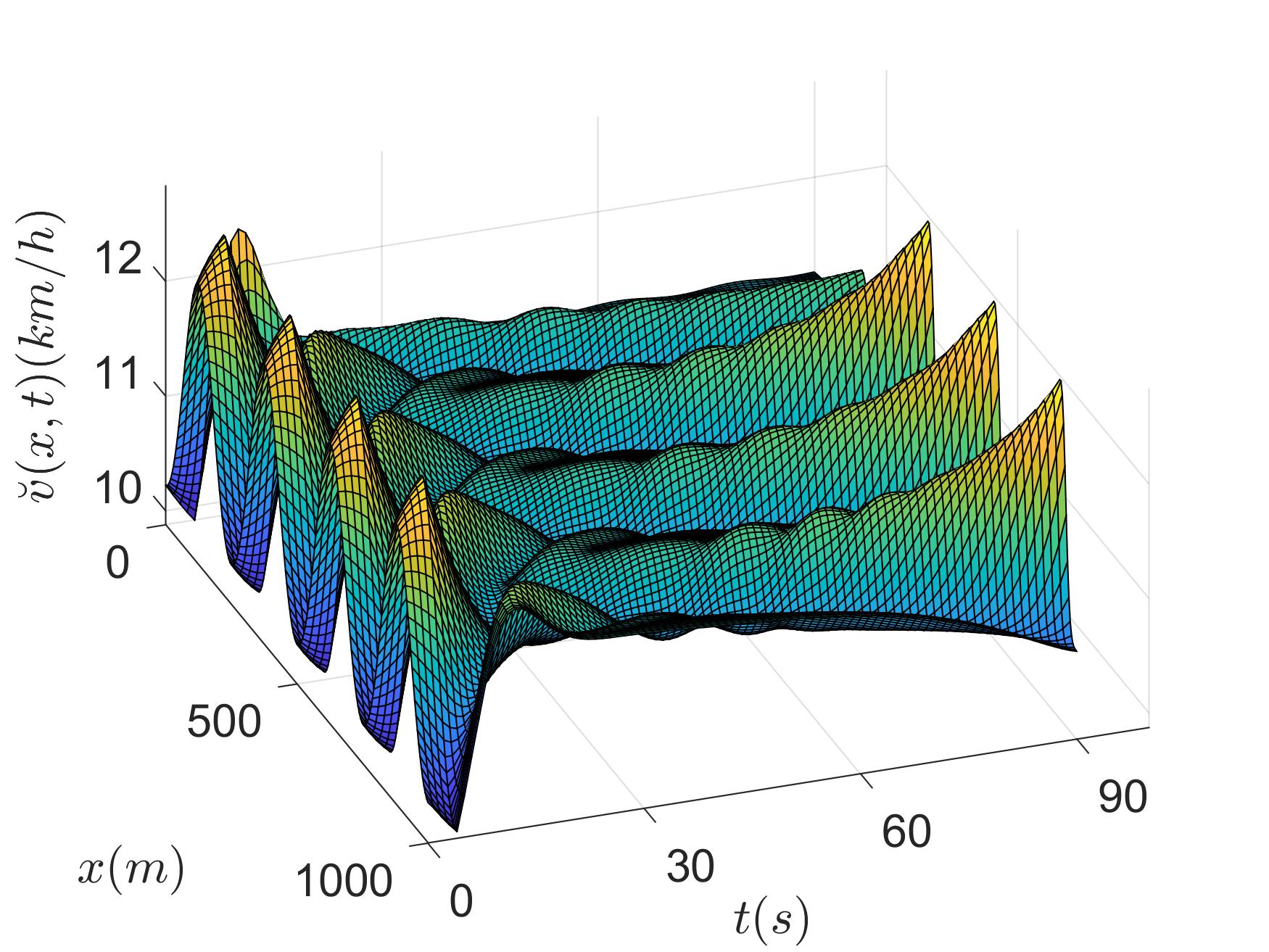}
                \end{minipage}
        }
        \subfigure[]{
                \begin{minipage}[ht]{0.3\linewidth}
                        \centering
                        \includegraphics[width=1.0\linewidth]{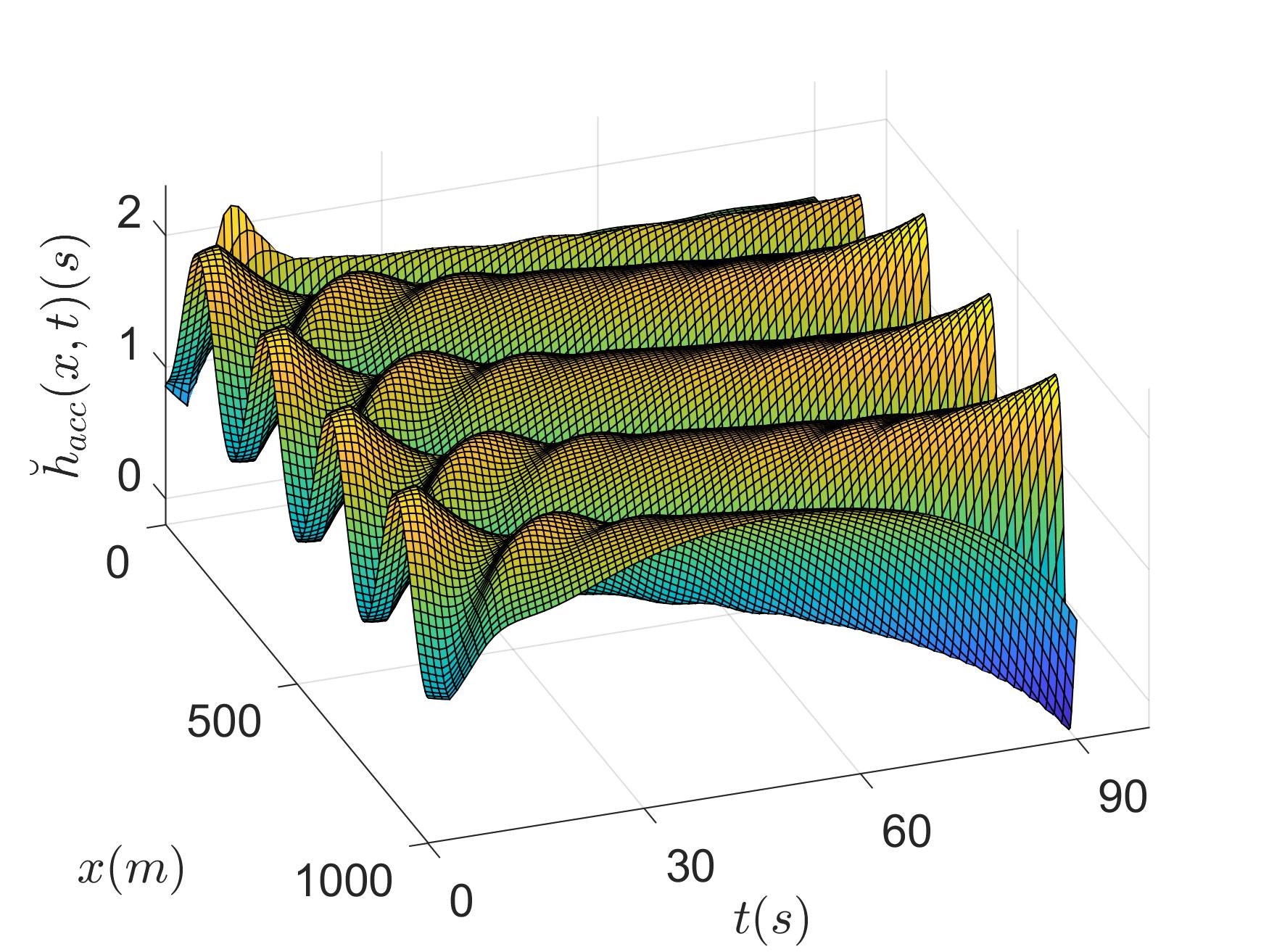}
                \end{minipage}
        }
        
        \subfigure[]{
                \begin{minipage}[ht]{0.3\linewidth}
                        \centering
                        \includegraphics[width=1.0\linewidth]{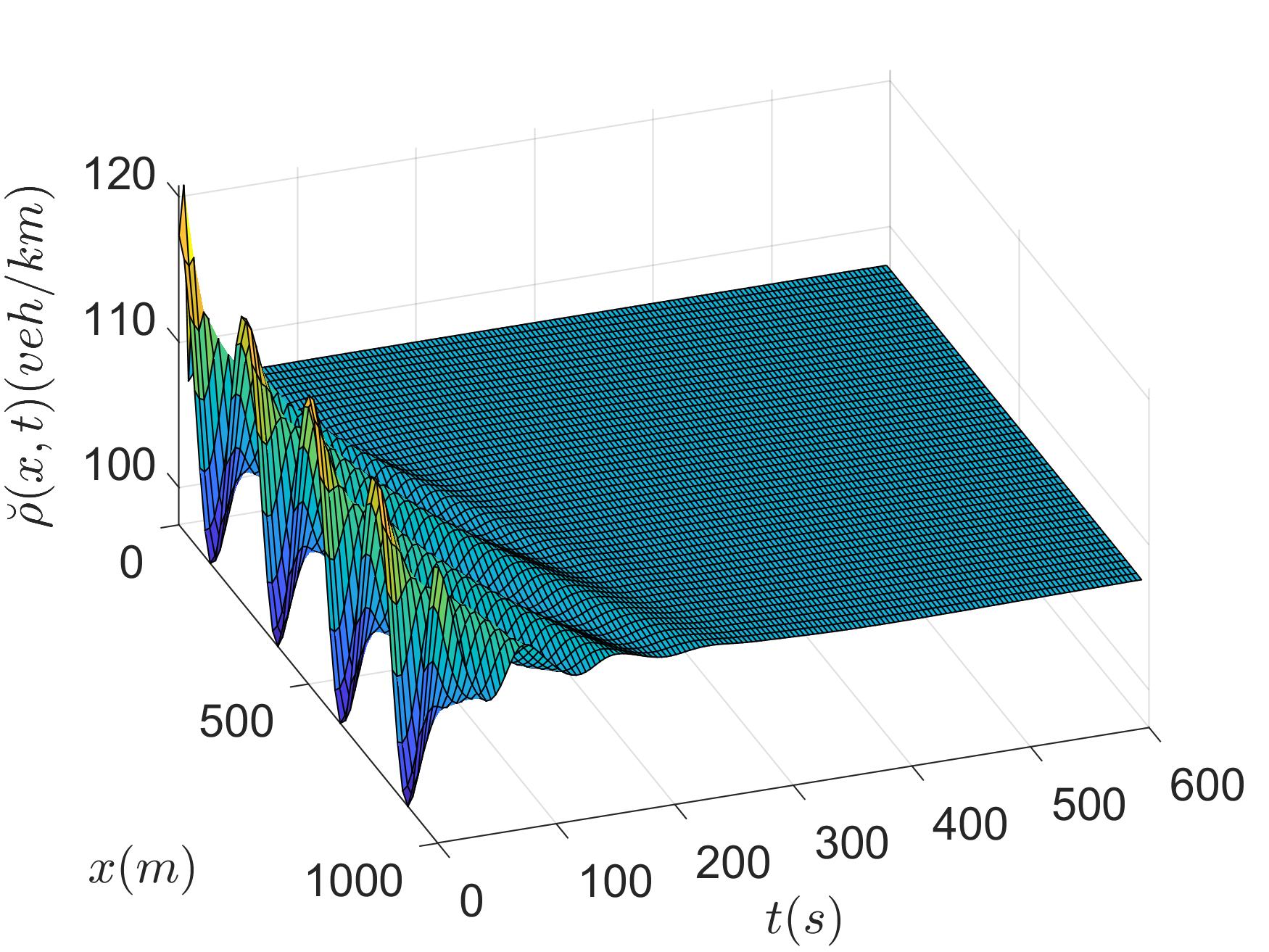}
                \end{minipage}
        }
        \subfigure[]{
                \begin{minipage}[ht]{0.3\linewidth}
                        \centering
                        \includegraphics[width=1.0\linewidth]{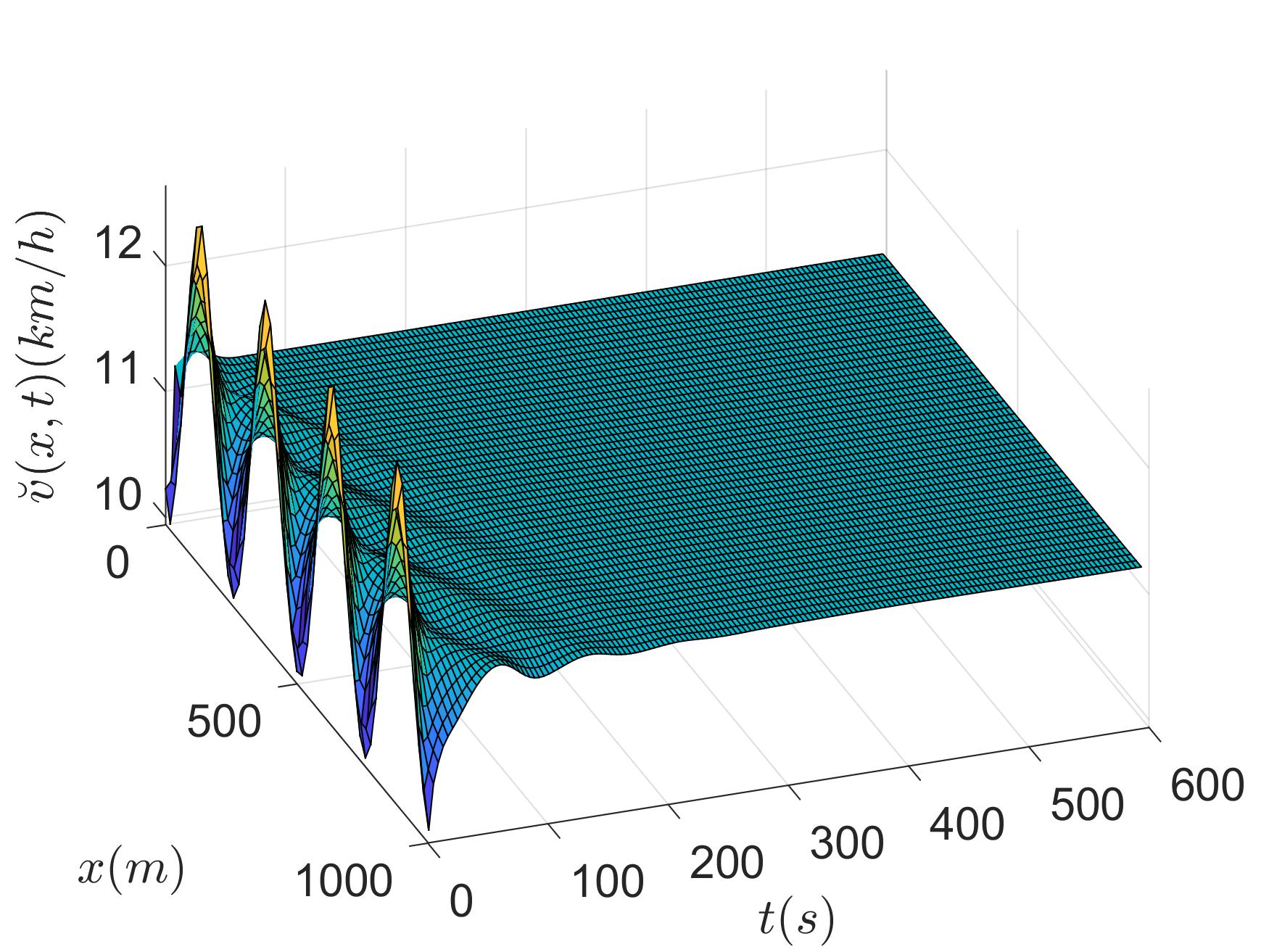}
                \end{minipage}
        }
        \subfigure[]{
                \begin{minipage}[ht]{0.3\linewidth}
                        \centering
                        \includegraphics[width=1.0\linewidth]{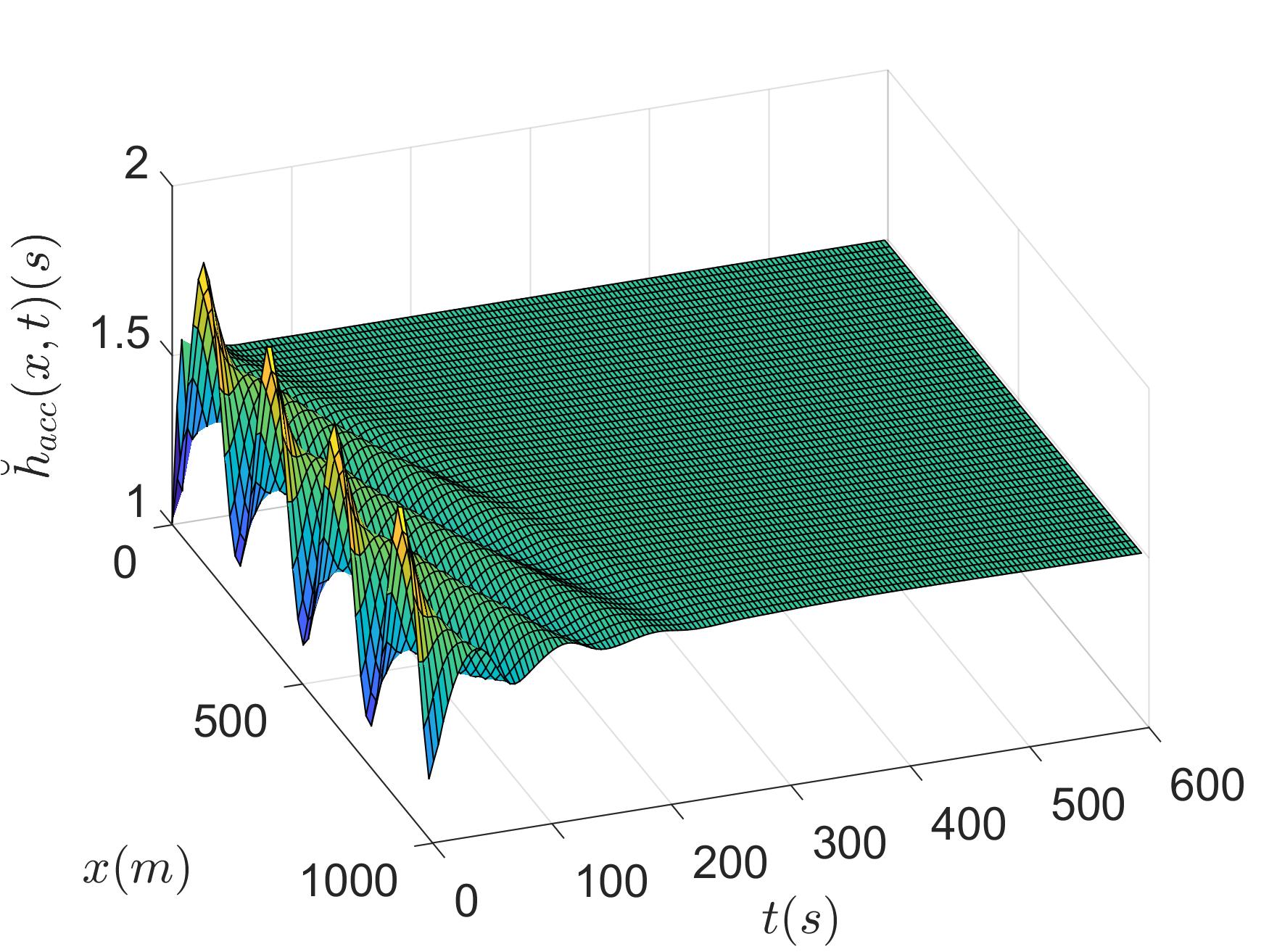}
                \end{minipage}
        }
        
        \subfigure[]{
                \begin{minipage}[ht]{0.3\linewidth}
                        \centering
                        \includegraphics[width=1.0\linewidth]{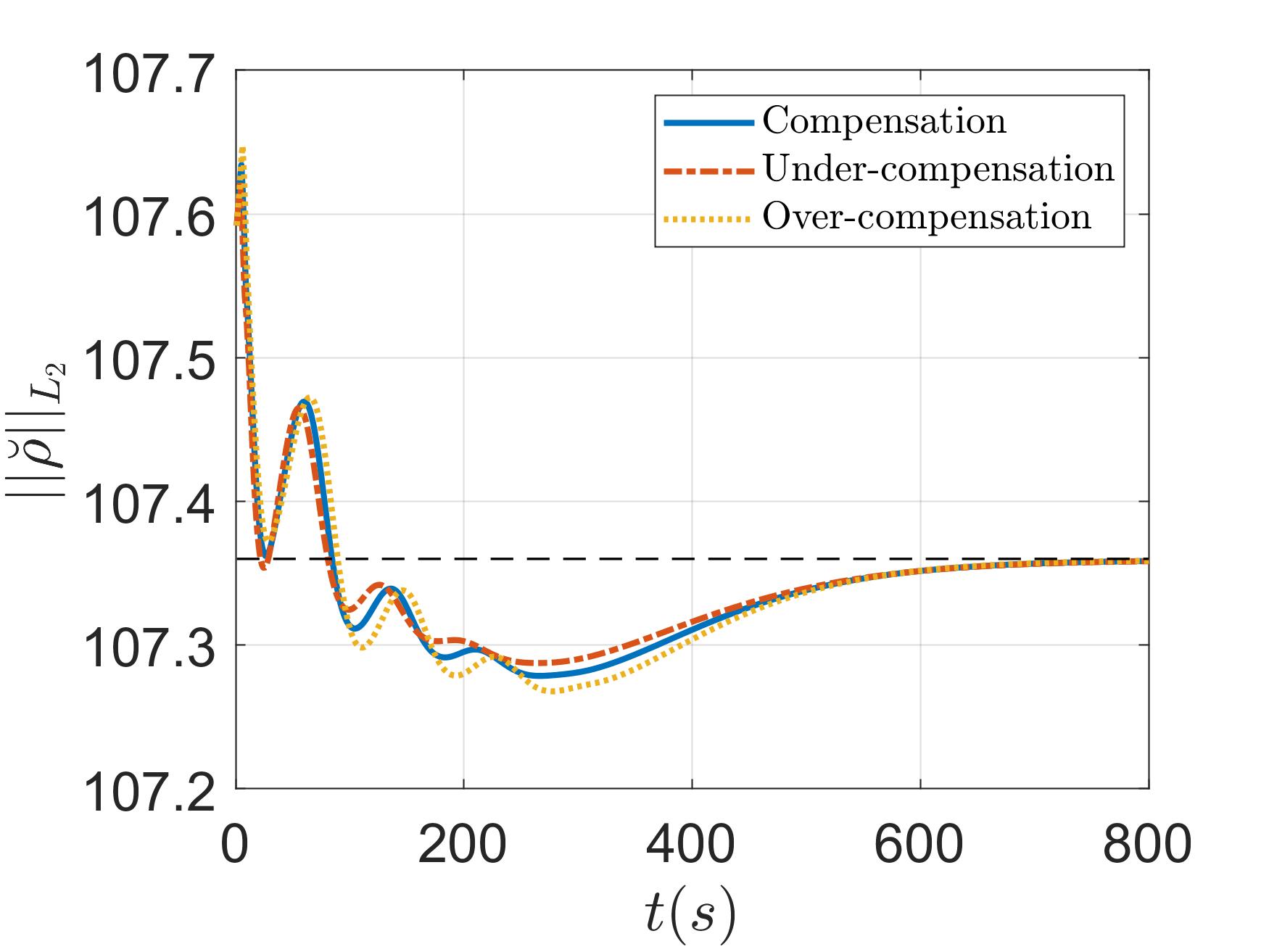}
                \end{minipage}
        }
        \subfigure[]{
                \begin{minipage}[ht]{0.3\linewidth}
                        \centering
                        \includegraphics[width=1.0\linewidth]{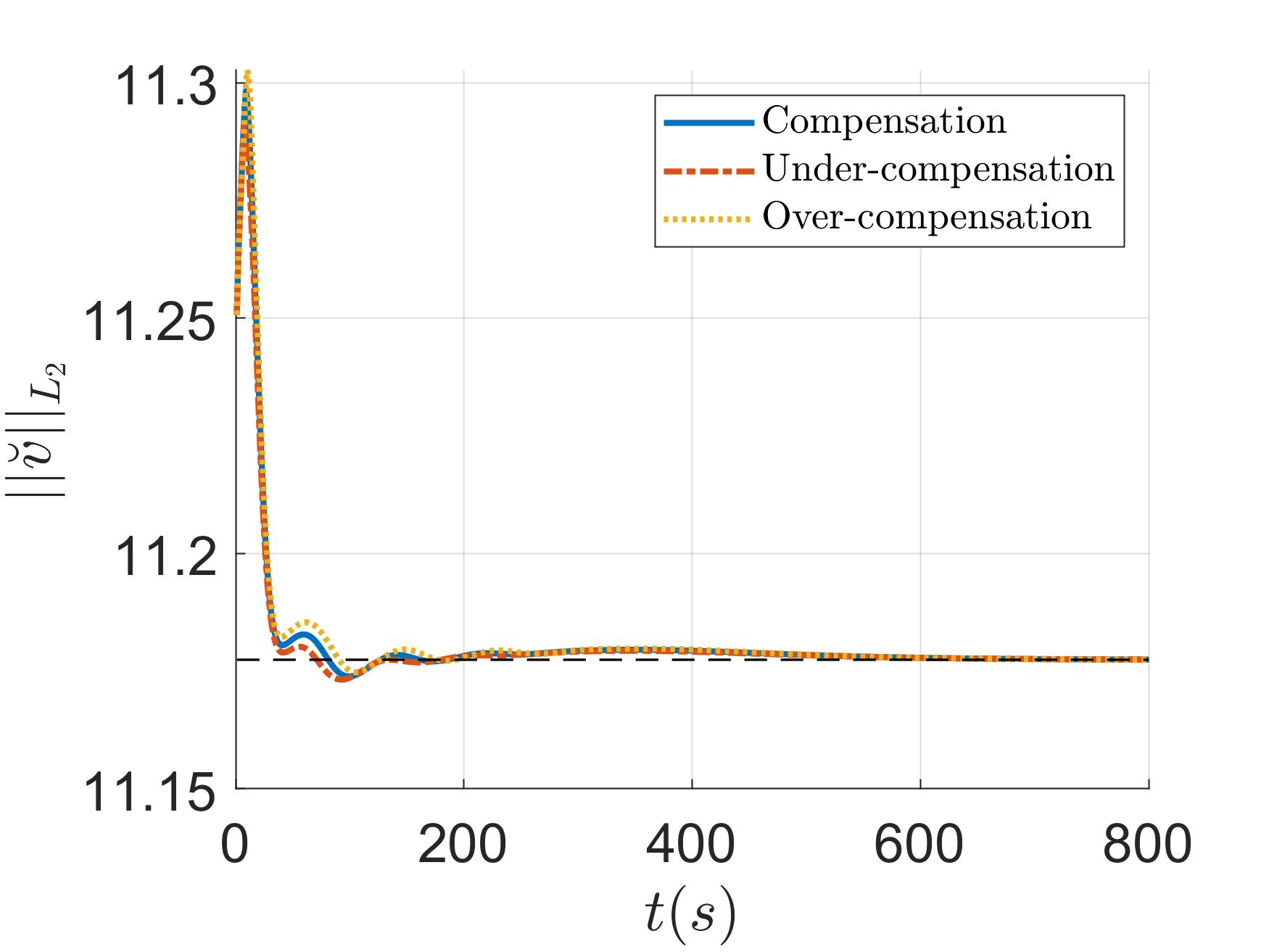}
                \end{minipage}
        }
        \subfigure[]{
                \begin{minipage}[ht]{0.3\linewidth}
                        \centering
                        \includegraphics[width=1.0\linewidth]{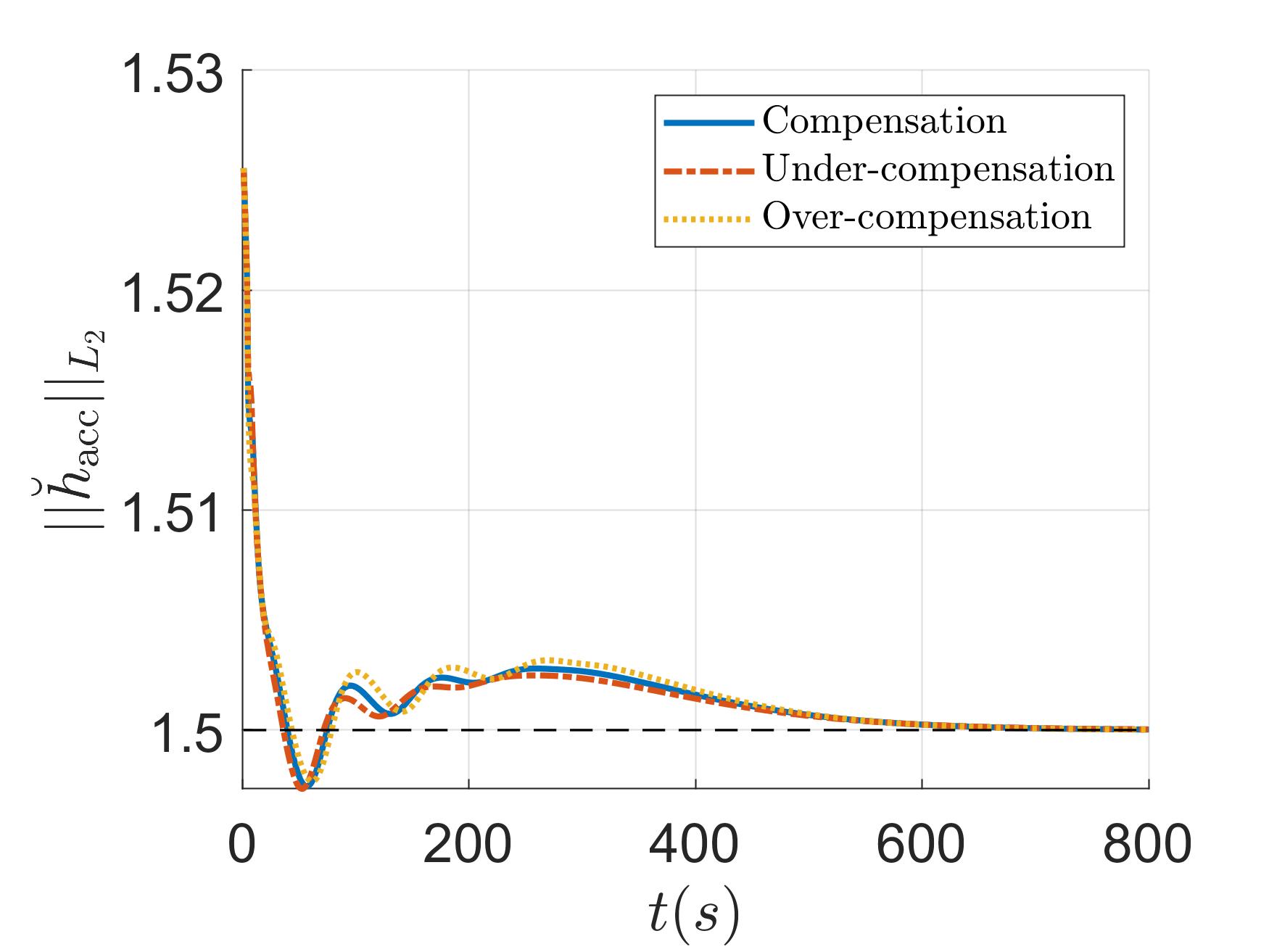}
                \end{minipage}
        }
        
        \caption{The evolution of the traffic states and the actuator: (a) $ \breve \rho(x,t)$ without delay-compensation; (b)  $ \breve v(x,t)$ without delay-compensation; (c) control effort $\breve  h_{\rm{acc}}(x,t)$ with no delay-compensation; (d)  $ \breve \rho(x,t)$ with delay-compensator; (e) state $\breve v(x,t)$ with delay-compensator; (f) control effort $ \breve h_{\rm{acc}}(x,t)$ with delay-compensator; (g) the $L_2$ norm of state $ \breve \rho(x,t)$ for delay matched and delay mismatched  cases; (h) the $L_2$ norm of state $\breve v(x,t)$  for delay matched and delay mismatched  cases; (i) the $L_2$ norm of control $\breve h_{\rm{acc}}(x,t)$  for delay matched and delay mismatched  cases.} 
        \label{fig:1}
\end{figure*}

\begin{lemma}\label{th-norm-V}
\rm{   The Lyapunov functions $V_1$ defined in \eqref{equ-V1} 
and $V_2$ defined in \eqref{equ-V2} are equivalent in the sense of the  $L_2$ norm, i.e., there exist positive constants $\alpha_1$ and $\alpha_2$ , such that       
\begin{align}
\alpha_1V_{2}(t)\leq V_{1}(t)\leq \alpha_2 V_2(t).\label{inequ-V}
\end{align}}            
\end{lemma}
Since the transformation \eqref{trans-beta} and the inverse transformation \eqref{inverse-psi} are presented in explicit form and they are bounded,  the $L_2$ norm equivalence between  the Lyapunov functions $V_1$ and $V_2$ is  easily established from  Theorem 1.2 \cite{anfinsen2019adaptive}. 
Combining Lemma \ref{lemma-V0}-\ref{th-norm-V}, we  reach Theorem  \ref{th-hacc}.

\section{Simulation}\label{simulation}
%
In this section, we  illustrate our results with   an numerical  example.  
We apply   the control law  \eqref{control-hacc} directly on the nonlinear model   \eqref{traffic_1}-\eqref{traffic_4}. The parameters  utilized in the simulation are shown in   Table \ref{table-1}, for which we choose the  same values as those in    \cite{9018188}.   


\begin{table}[ht]
\centering
\caption{PARAMETERS OF SYSTEM}\label{table-1}
\begin{tabular}{|l|l|l|}
\hline
$L=1000~\rm{m}$ & $l=5~\rm{m}$ & $q_{\rm{in}} = 1200~ \rm{veh}/h$ \\
$\tau_{\rm{acc}}=2~\rm{s}$ & $\tau_{\rm{m}} = 60~\rm{s}$ & $\breve h_{\rm{m}} = 1~\rm{s}$\\
$\bar{h}_{\rm{acc}}=1.5~\rm{s}$ & $\breve \rho_{\rm{min}}=37\rm{veh}/km$ & \\
\hline
\end{tabular}
\end{table}
From  \eqref{equ-hmix},  we get the value of the mixed steady-state time gap  $\bar h_{\rm{mix}}=1. 35~\rm{ s}$.      
The steady-state values for density and speed  derived from \eqref{steady-state} are $ \bar{\rho}=107.36~\rm{ veh/km }$, $ \bar{v}=11.18~\rm{km/h} $. 
The system is discretized   with time step $\Delta t = 0.5~\rm{s} $ and spatial step $\Delta x = 5~\rm{m}$. The initial conditions are chosen as  $ \breve \rho(x,0)=10\cos(8\pi x/L) $ and $ \breve v(x,0)=q_{\rm{in}}/\breve \rho(x,0) $, which imitates the stop-and-go wave in congested regime. Considering the time delay  $D=4~{\rm{s}}$ and the coefficient $k=0.1(1/\rm{s})$ of the target system, we first investigate the  numerical solution of the nonlinear system \eqref{traffic_1}-\eqref{traffic_4} with the non-delay-compensation state-feedback control proposed in \cite{9018188}, whose results are shown in Fig. \ref{fig:1} (a)-(c). It is evident that the response of the system without delay-compensation exhibits an unstable and  oscillatory behavior and the stop-and-go wave propagates backward without being attenuated.     
In contrast, as it is shown in Fig. \ref{fig:1} (d) and (e), the traffic system  \eqref{traffic_1}-\eqref{traffic_4}  is stabilized under the delay-compensator and the oscillations in the speed response are suppressed first and then the oscillations in density response converges to the equilibrium as well.  The control effort \eqref{control-hacc} is shown in Fig. \ref{fig:1} (f), which indicates the resulting values for the time gap of ACC vehicles lie within the  interval $[0.8, 2.2]~\rm{s}$ which is typically implemented in ACC vehicles settings, see, e.g., \cite{nowakowski2011cooperative}. In order to illustrate the converging evolution of the states and the actuator more clearly, we plot the  states and the control effort in $L_2$ norm in solid line, respectively, shown in Fig. \ref{fig:1} (g)-(i), where the dashed line represents the steady-state value.  

To examine the robustness of the proposed delay-compensated control law, we conducted some additional simulations with mismatched delays. First, we consider the over-compensation case where the actual delay is less than the delay used in control. The numerical result is shown in Fig.  \ref{fig:1} (g)-(i) in dotted line when the actual delay is $3 \rm{s}$, while the delay used in control is $4 \rm{s}$.  Then, we consider the under-compensation case where the actual delay is larger than the delay used in control. The numerical result is also  shown in Fig.  \ref{fig:1} (g)-(i),  where the  dash-dotted line represents $L_2$ norm of the states $\breve \rho$, $\breve v$ and the control effort, respectively, when the actual delay is $5 \rm{s}$, while the delay used in control is $4 \rm{s}$.  From the simulation results, we find the closed-loop system remains stable in  delay mismatched conditions. 
          The performance
improvement of the closed-loop system under the proposed
controller \eqref{control-hacc} is also  illustrated in simulation by employing three
 metrics, namely,  total travel time (TTT), 
fuel consumption and comfort, where we use the same definition as in  \cite{treiber2014traffic}:
\begin{align}
        J_{\mathrm{TTT}} &  = \int_{0}^{T} \int_{0}^{L} \rho(x, t) d x d t, \\
        J_{\text {fuel }} &  = \int_{0}^{T} \int_{0}^{L} \xi(x, t) \cdot \rho(x, t) d x d t,\\
        J_{\text {comfort }} & = \int_{0}^{T} \int_{0}^{L}\left(a(x, t)^{2}+a_{t}(x, t)^{2}\right) \rho(x, t) d x d t,
\end{align}
where we also select the same functions and parameters as those in  \cite{treiber2014traffic}:  $a(x,t) = v_{t}(x, t)+v(x, t) v_{x}(x, t)$,
$b_0=25\cdot10^{-3}$,
$b_1=24.5\cdot10^{-6}$,
$b_3=32.5\cdot10^{-9}$,
$b_4=125\cdot10^{-6}$,
$\xi(x, t)=\max \left\{0, b_{0}+b_{1} v(x, t)+b_{3} v(x, t)+b_{4} v(x, t) a(x, t)\right\} $
and $T=300~\rm{s}$.

\begin{table}[ht]
        \centering
        \caption{PERFORMANCE INDICES}\label{table-2}
        \setlength{\tabcolsep}{5mm}{
                \begin{tabular}{|c|c|}
                        \hline
                        Performance indices& Percentage with improvement \\
                        \hline
                        $J_{\mathrm{TTT}}$& $3.91\%$\\\hline
                        $J_{\text {fuel }}$& $3.76\%$\\\hline
                        $J_{\text {comfort }}$&$92.1\%$ \\\hline
        \end{tabular}}
\end{table}
 As shown in the Table \ref{table-2}, three indicators are improved compared to the open-loop system. Since the open-loop system is unstable, we only consider  the first 300 seconds in the simulation.  In particular, the driving comfort is significantly improved  due to  the  homogenization of speed which alleviates the phenomenon of stop-and-go oscillation.  

\section{Conclusion}\label{conclusion}
In this paper, we present a  control design methodology for compensation of  unstable traffic flow with input delay. The delay is resulting from the time required for the transmission of the control command from the control to the ACC vehicles.    Applying the PDE backstepping method, we  develop an explicit   feedback delay-compensator, composed of the feedback of the traffic speed, the traffic density and the historical actuator states, which is divided  into three parts upon  the spatial domain along the highway. 
The closed-loop system, under the developed compensator, was shown to be exponentially stable in the $L_2$ sense. Although the control design is based on a linearized system, the numerical simulation shows the effectiveness  of the  proposed controller on the original nonlinear system.    
Further research would  include delay-robustness and the output feedback control based on observer design.

\appendices

\section{Solving kernel equations}\label{kernel-solution}

        Under Assumption \ref{assump-1}, we get the integral equations from \eqref{equ-gamma}--\eqref{bnd-eta2} by using the characteristic line method, 
\begin{align}
        \gamma(x,s,y)=&-\frac{c_5}{c_6}\mathrm{e}^{-{c}_2x}\delta(x-y-c_1s)\nonumber\\
        &-\int_{y}^{y+c_1s}\frac{c_5}{c_1}\mathrm{e}^{-{c}_2\theta}\eta(x,\frac{y-\theta}{c_1}+s,\theta)d\theta,\label{gamma-1}\\
        &~~\mathrm{if}~~y\le L-c_1s,\nonumber\\
        \gamma(x,s,y)=&-\int^{s-\frac{L-y}{c_1}}_{0}\frac{c_4c_5}{c_1}\mathrm{e}^{-{c}_2L}\eta(x,\theta,L)d\theta\nonumber\\
        &-\int_{y}^{L}\frac{c_5}{c_1}\mathrm{e}^{-{c}_2\theta}\eta(x,\frac{y-\theta}{c_1}+s,\theta)d\theta,\label{gamma-2}\\
        &~~\mathrm{if}~~y> L-c_1s,\nonumber\\
        \eta(x,s,y)=&-\frac{c_1c_7}{c_4}\gamma(x,s-\frac{y}{c_4},0),\label{eta-1}\\
        &~~\mathrm{if}~~y\le c_4s,\nonumber\\
        \eta(x,s,y)=&\frac{k}{c_6}\delta(x-(y-c_4s)),\label{eta-2}\\
        &~~\mathrm{if}~~y> c_4s.\nonumber
\end{align}
Substitute \eqref{eta-1} and \eqref{eta-2} into \eqref{gamma-1} and \eqref{gamma-2}, which gives three-branch expression of $\gamma(x,s,y)$
\begin{align}
        \gamma(x,s,y)=
        &-\int_{0}^{y+c_1s}\frac{kc_5}{c_6(c_1+c_4)}\mathrm{e}^{-\frac{c_1c_2}{c_1+c_4}\theta-\frac{c_2c_4}{c_1+c_4}(y+c_1s)}\nonumber\\
        &\times\delta(x-\theta)d\theta\nonumber\\
        &+\int_{0}^{s-\frac{y}{c_4}}\frac{c_1c_5c_7}{c_1+c_4}\mathrm{e}^{\frac{c_2c_4}{c_1+c_4}(c_1\theta-y-c_1s)}\gamma(x,\theta,0)d\theta\nonumber\\
        &-\frac{c_5}{c_6}\mathrm{e}^{-{c}_2x}\delta(x-y-c_1s),\label{gamma-iteration-1}\\
        &~~\mathrm{if}~~0\le y\le c_4s,\nonumber\\
        \gamma(x,s,y)=
        &-\int_{y-c_4s
        }^{y+c_1s
        }\frac{kc_5}{c_6(c_1+c_4)}\mathrm{e}^{-\frac{c_1c_2}{c_1+c_4}\theta-\frac{c_2c_4}{c_1+c_4}(y+c_1s)}\nonumber\\
        &\times\delta(x-\theta)d\theta\nonumber\\
        &-\frac{c_5}{c_6}\mathrm{e}^{-{c}_2x}\delta(x-y-c_1s),\label{gamma-iteration-2}\\
        &~~\mathrm{if}~~c_4s<y\le L-c_1s,\nonumber\\
        \gamma(x,s,y)=
        &-\int_{y-c_4s
        }^{\frac{c_1+c_4}{c_1}L-\frac{c_4}{c_1}y-c_4s
        }\frac{kc_5}{c_6(c_1+c_4)}\nonumber\\
&\times\mathrm{e}^{-\frac{c_1c_2}{c_1+c_4}\theta-\frac{c_2c_4}{c_1+c_4}(y+c_1s)}\delta(x-\theta)d\theta\nonumber\\
        &-\frac{kc_5}{c_1c_6}\mathrm{e}^{-{c}_2L},\label{gamma-iteration-3}\\
        &~~\mathrm{if}~~L-c_1s<y\le L\nonumber.
\end{align}
We use successive approximations method for  \eqref{gamma-iteration-1} and get the following iterations: 
\begin{align}
&\gamma^{n+1}(x,s,y)=
        -\frac{kc_5}{c_6(c_1+c_4)}\mathrm{e}^{-\frac{c_1c_2}{c_1+c_4}x-\frac{c_2c_4}{c_1+c_4}(y+c_1s)}\nonumber\\
        &-\frac{c_5}{c_6}\mathrm{e}^{-{c}_2x}\delta(x-y-c_1s)\nonumber\\
        &+\int_{0}^{s-\frac{y}{c_4}}\frac{c_1c_5c_7}{c_1+c_4}\mathrm{e}^{\frac{c_2c_4}{c_1+c_4}(c_1\theta-y-c_1s)}\gamma^{n}(x,\theta,0)d\theta,\nonumber\\
   & ~~~~~~~\mathrm{if} ~~~0\leq x\leq y+c_1s, ~~\mathrm{for} ~n=0,1,2,...  
\end{align}
 Let 
\begin{align}
         \Delta\gamma^n(x,s,y)=&\gamma^{n+1}(x,s,y)-\gamma^{n}(x,s,y),
      \end{align}
      then,
      \begin{align}
       \Delta\gamma^0(x,s,y)=
        &-\frac{kc_5}{c_6(c_1+c_4)}\mathrm{e}^{-\frac{c_1c_2}{c_1+c_4}x-\frac{c_2c_4}{c_1+c_4}(y+c_1s)}\nonumber\\
        &-\frac{c_5}{c_6}\mathrm{e}^{-{c}_2x}\delta(x-y-c_1s),\\\Delta\gamma^n(x,s,y)=&\int_{\frac{x}{c_1}}^{s-\frac{y}{c_4}}\frac{c_1c_5c_7}{c_1+c_4}\mathrm{e}^{\frac{c_2c_4}{c_1+c_4}(c_1\theta-y-c_1s)}\nonumber\\
        &\times\Delta\gamma^{n-1}(x,\theta,0)d\theta, 
        \end{align}
 After a series of  iterations of $\gamma^n$, we have 
\begin{align}
        &\Delta\gamma^n=
        -\frac{c_5}{c_6} \mathrm{e}^{-\frac{c_1c_2}{c_1+c_4}x-\frac{c_2c_4}{c_1+c_4}(y+c_1s)}\times\nonumber  \\
 &\left[\frac{c_5c_7(c_1c_5c_7(s-\frac{y}{c_4}-\frac{x}{c_1}))^{n-1}}{(n-1)!(c_1+c_4)^n}
        -\frac{k(c_1c_5c_7(s-\frac{y}{c_4}-\frac{x}{c_1}))^n}{n!(c_1+c_4)^{n+1}}\right]\nonumber\\
        &~~~~~~~~~~~~~~~\mathrm{if}~~0\le y\le c_4s-\frac{c_4}{c_1}x,
     \end{align}
 then,
 \begin{align}      
        \gamma&(x,s,y)=\sum_{n=1}^{+\infty}\Delta\gamma^n(x,s,y)+\Delta\gamma^0(x,s,y)\nonumber\\
        =&-\frac{c_5(k+c_5c_7)}{c_6(c_1+c_4)}\mathrm{e}^{-c_2(x+y)} 
        -\frac{c_5}{c_6}\mathrm{e}^{-{c}_2x}\delta(x-y-c_1s), \nonumber\\ 
        &~~~~~\mathrm{if}~~0\le y\le c_4s-\frac{c_4}{c_1}x.
\end{align}
After a direct computing from \eqref{gamma-iteration-2} and \eqref{gamma-iteration-3}, we reach \eqref{kernel-gamma}.   

\section{The Proof of Lemma \ref{trans-beta-th}}\label{tran-proof}        
First, we prove that the transformation is continuous. 
Substitute $x=c_1s$ into the first and the second case of the transformation \eqref{trans-beta}, which gives
\begin{align}\label{equ-1-bnd}
&T_1[\psi(t)](c_1s,s)+Z_1[z(t)](c_1s)+Y_1[v(t)](c_1s)\nonumber
\\=&\psi(c_{1}s,s,t)\nonumber\\
&-\int_{0}^s c_1c_2\mathrm{e}^{-c_1c_2\tau}\psi(c_{1}s-c_{1}\tau,s-\tau,t)d\tau\nonumber\\
&+\int_{0}^s k\psi(c_{1}s+c_4\tau,s-\tau,t)d\tau\nonumber\\
&-\int_{0}^s \int_{c_{1}s-c_1 \tau}^{c_{1}s+c_{4}\tau} \frac{kc_1c_2}{c_1+c_4}\mathrm{e}^{-\frac{c_1c_2}{c_1+c_4}(c_{1}s-y+c_4\tau)}\nonumber\\
&\times\psi(y,s-\tau,t)dyd\tau\nonumber 
\\&+\int_{0}^{c_{1}s+c_4s}\frac{kc_5}{c_6(c_1+c_4)}\mathrm{e}^{-c_{1}c_2s-\frac{c_2c_4}{c_1+c_4}y}z(y,t)dy
\nonumber  \\&-\frac{c_5c_7}{c_6}\mathrm{e}^{-c_1c_2s}v(0,t) -\frac{k}{c_6}v(c_1s+c_4s,t),
\end{align}
and 
\begin{align}\label{equ-2-bnd}
&T_2[\psi(t)](c_1s,s)+Z_2[z(t)](c_1s)+Y_2[v(t)](c_1s)\nonumber
\\=&\psi(c_{1}s,s,t)\nonumber\\
&-\int_{0}^s c_1c_2\mathrm{e}^{-c_1c_2\tau}\psi(c_{1}s-c_{1}\tau,s-\tau,t)d\tau\nonumber\\
&+\int_{0}^s k\psi(c_{1}s+c_4\tau,s-\tau,t)d\tau\nonumber\\
&-\int_{0}^s \int_{c_{1}s-c_1 \tau}^{c_{1}s+c_{4}\tau} \frac{kc_1c_2}{c_1+c_4}\mathrm{e}^{-\frac{c_1c_2}{c_1+c_4}(c_{1}s-y+c_4\tau)}\nonumber\\
&\times\psi(y,s-\tau,t)dyd\tau\nonumber 
\\&+\frac{c_5}{c_6}\mathrm{e}^{-c_1c_2s}z(0,t)\nonumber \\&+\int_{0}^{c_{1}s+c_4s}\frac{kc_5}{c_6(c_1+c_4)}\mathrm{e}^{-c_{1}c_2s-\frac{c_2c_4}{c_1+c_4}y}z(y,t)dy
\nonumber  \\& -\frac{k}{c_6}v(c_1s+c_4s,t),
\end{align}
It shows that \eqref{equ-1-bnd} equals to \eqref{equ-2-bnd} by using the  the boundary condition \eqref{equ-bnd-z}.  In a similar way, one  can get   
\begin{align*}
T_2[\psi(t)](L-c_4s,s)=&T_3[\psi(t)](L-c_4s,s),\\
Z_2[z(t)](L-c_4s)=& Z_3[z(t)](L-c_4s),\\
Y_2[v(t)](L-c_4s,s)=& \frac{k}{c_6}v(L,t), 
\end{align*}
which implies the transformation  \eqref{trans-beta} in the second and third case has a same value at $x=L-c_4s$.

Second, it is obvious that  the transformation is bounded from the explicit form of  each kernel function. 

Third, differentiating the transformation \eqref{trans-beta} with respect to $t$ and $s$, respectively, and then substituting them into \eqref{equ-beta} and \eqref{equ-bnd-beta}, one can prove the  system \eqref{equ-z}-\eqref{equ-bnd-psi} can be transformed into  \eqref{equ-z1}-\eqref{equ-bnd-beta} via \eqref{trans-beta} after a lengthy computation. The details are shown as follows: 
in the case of  $0\le x < c_1s $,
\begin{align}
&\beta_s(x,s,t)=\psi_s(x,s,t)\nonumber\\
&~~~+\frac{c_4c_5(k+c_5c_7)}{c_6(c_1+c_4)}\mathrm{e}^{-c_2(c_4s+\frac{c_1-c_4}{c_1}x)}z(c_4(s-\frac{x}{c_1}),t)\nonumber\\
&~~~+\frac{kc_4c_5}{c_6(c_1+c_4)}\mathrm{e}^{-c_2(x+c_4s)}z(x+c_4s,t)\nonumber\\
&~~~-\frac{kc_4c_5}{c_6(c_1+c_4)}\mathrm{e}^{\frac{c_2(c_4-c_1)}{c_1}x-c_2c_4s}z(c_4(s-\frac{x}{c_1}),t)\nonumber\\ 
&~~~-\int_{c_4(s-\frac{x}{c_1})}^{x+c_4s}\frac{kc_1c_2c_4c_5}{c_6(c_1+c_4)^2}\mathrm{e}^{-\frac{c_1c_2}{c_1+c_4}x-\frac{c_2c_4}{c_1+c_4}(y+c_1s)}\nonumber\\
&~~~\times z(y,t)dy\nonumber\\      
&~~~-\frac{c_4c_5c_7}{c_6}\mathrm{e}^{-c_2x}v_x(c_{4}(s-\frac{x}{c_1}) ,t)\nonumber\\
&~~~-\frac{kc_4}{c_6}v_x(x+c_4s,t)\nonumber\\
&~~~-\frac{c_1c_5c_7(k+c_5c_7)}{c_6(c_1+c_4)}\mathrm{e}^{-c_2x} v(c_4(s-\frac{x}{c_1}),t)\nonumber\\
&~~~-\int_{0}^{\frac{x}{c_1}}c_1c_2\mathrm{e}^{-c_1c_2r}\psi_s(x-c_{1}r,s-r,t)dr\nonumber\\
&~~~+\int_{\frac{x}{c_1}}^sc_5c_7\mathrm{e}^{-c_2x}\psi_s(c_{4}(r-\frac{x}{c_1}),s-r,t)dr\nonumber\\
&~~~+c_5c_7\mathrm{e}^{-c_2x}\psi(c_{4}(s-\frac{x}{c_1}),0,t)\nonumber\\  
&~~~+\int_{0}^s k\psi_s(x+c_4r,s-r,t)dr\nonumber\\  
&~~~+ k\psi(x+c_4s,0,t)\nonumber\\
&~~~-\int_{0}^s \int_{\max\{x-c_1r,c_{4}(r-\frac{x}{c_1})\}}^{x+c_{4}r} \frac{kc_1c_2}{c_1+c_4}\mathrm{e}^{-\frac{c_1c_2}{c_1+c_4}(x-y+c_4r)}\nonumber\\
&~~~\times\psi_s(y,s-r,t)dydr\nonumber\\
&~~~-\int_{c_{4}(s-\frac{x}{c_1})}^{x+c_{4}s} \frac{kc_1c_2}{c_1+c_4}
\mathrm{e}^{-\frac{c_1c_2}{c_1+c_4}(x-y+c_4s)}\psi(y,0,t)dydr;
\end{align}
in the case of $c_1s < x \leq L-c_4s $,
\begin{align}
&\beta_s(x,s,t)=\psi_s(x,s,t)\nonumber\\
&~~~-\frac{c_1c_5}{c_6}\mathrm{e}^{-c_2x}z_x(x-c_1s,t)\nonumber\\
&~~~+\frac{kc_4c_5}{c_6(c_1+c_4)}\mathrm{e}^{-c_2(x+c_4s)}z(x+c_4s,t)\nonumber\\
&~~~+\frac{kc_1c_5}{c_6(c_1+c_4)}\mathrm{e}^{-c_2x}z(x-c_1s,t)\nonumber\\
&~~~-\int_{x-c_1s}^{x+c_4s}\frac{kc_1c_2c_4c_5}{c_6(c_1+c_4)^2}\mathrm{e}^{-\frac{c_1c_2}{c_1+c_4}x-\frac{c_2c_4}{c_1+c_4}(y+c_1s)}z(y,t)dy\nonumber\\
&~~~-\frac{kc_4}{c_6}v_x(x+c_4s,t)\nonumber\\
&~~~-\int_0^s c_1c_2\mathrm{e}^{-c_1c_2r} \psi_s(x-c_1r,s-r,t)dr\nonumber\\
&~~~-c_1c_2\mathrm{e}^{-c_1c_2s} \psi(x-c_1s,0,t)\nonumber\\
&~~~+\int_0^s k \psi_s(x+c_4r,s-r,t)dr\nonumber\\
&~~~+k \psi(x+c_4s,0,t)\nonumber\\
&~~~-\int_0^s \int_{x-c_1r}^{x+c_4r}\frac{kc_1c_2}{c_1+c_4}\mathrm{e}^{-\frac{c_1c_2}{c_1+c_4}(x-y+c_4r)} \nonumber\\
&~~~\times\psi_s(y,s-r,t)dydr\nonumber\\
&~~~-\int_{x-c_1s}^{x+c_4s}\frac{kc_1c_2}{c_1+c_4}\mathrm{e}^{-\frac{c_1c_2}{c_1+c_4}(x-y+c_4s)} \nonumber\\
&~~~\times\psi(y,0,t)dy; 
\end{align}
and in the case of  $L-c_4s < x \leq L$,
\begin{align}
&\beta_s(x,s,t)=\psi_s(x,s,t)\nonumber\\
&~~~-\frac{c_1c_5}{c_6}\mathrm{e}^{-c_2x}z_x(x-c_1s,t)\nonumber\\
&~~~-\frac{kc_1c_5}{c_6(c_1+c_4)}\mathrm{e}^{-c_2L}z(\frac{c_1+c_4}{c_4}L-\frac{c_1}{c_4}x-c_1s,t)\nonumber\\
&~~~+\frac{kc_1c_5}{c_6(c_1+c_4)}\mathrm{e}^{-c_2x}z(x-c_1s,t)\nonumber\\
&~~~-\int_{x-c_1s}^{\frac{c_1+c_4}{c_4}L-\frac{c_1}{c_4}x-c_1s}\frac{kc_1c_2c_4c_5}{c_6(c_1+c_4)^2}\nonumber\\
&~~~\times\mathrm{e}^{-\frac{c_1c_2}{c_1+c_4}x-\frac{c_2c_4}{c_1+c_4}(y+c_1s)}z(y,t)dy\nonumber\\
&~~~+\frac{kc_5}{c_6}\mathrm{e}^{c_2L}z(\frac{c_1+c_4}{c_4}L-\frac{c_1}{c_4}x-c_1s,t)\nonumber\\ 
&~~~-\int_0^s c_1c_2\mathrm{e}^{-c_1c_2r} \psi_s(x-c_1r,s-r,t)dr\nonumber\\
&~~~-c_1c_2\mathrm{e}^{-c_1c_2s} \psi(x-c_1s,0,t)\nonumber\\
&~~~+\int_0^{\frac{L-x}{c_4}} k \psi_s(x+c_4r,s-r,t)dr\nonumber\\
&~~~+\int_{\frac{L-x}{c_4}}^s k \psi_s(L,s-r,t)dr\nonumber\\
&~~~+ k \psi(L,0,t)dr\nonumber\\
&~~~-\int_0^s \int_{x-c_1r}^{\min\{x+c_4r,\frac{c_1+c_4}{c_4}L-\frac{c_1}{c_4}x-c_1r\}}\frac{kc_1c_2}{c_1+c_4}\nonumber\\
&~~~\times\mathrm{e}^{-\frac{c_1c_2}{c_1+c_4}(x-y+c_4r)} \psi_s(y,s-r,t)dydr\nonumber\\
&~~~-\int_{x-c_1s}^{\frac{c_1+c_4}{c_4}L-\frac{c_1}{c_4}x-c_1s}\frac{kc_1c_2}{c_1+c_4}\nonumber\\
&~~~\times\mathrm{e}^{-\frac{c_1c_2}{c_1+c_4}(x-y+c_4s)} \psi(y,0,t)dy\nonumber\\
&~~~-\int_{\frac{L-x}{c_4}}^s \int_{\frac{c_1+c_4}{c_4}L-\frac{c_1}{c_4}x-c_1r}^{L} kc_2\mathrm{e}^{-c_2L} \psi_s(y,s-r,t)dydr\nonumber\\
&~~~-\int_{\frac{c_1+c_4}{c_4}L-\frac{c_1}{c_4}x-c_1s}^{L} kc_2\mathrm{e}^{-c_2L} \psi(y,0,t)dy.\\
&~~~~~~~~\nonumber
\end{align}

Then, differentiating \eqref{trans-beta} with respect of $t$, substituting \eqref{equ-z}, \eqref{equ-v}, \eqref{equ-bnd-v} and \eqref{equ-psi}  to the time derivative of  \eqref{trans-beta}, and   using  integration by parts, we have $\beta_t$ in three cases as follows:
in the case of    $0\le x < c_1s$,
\begin{align}
&\beta_t(x,s,t)=\psi_t(x,s,t)\nonumber\\
&~~~-\frac{c_1c_5(k+c_5c_7)}{c_6(c_1+c_4)}\mathrm{e}^{-c_2(c_4s+\frac{c_1-c_4}{c_1}x)}z(c_4(s-\frac{x}{c_1}),t)\nonumber\\
&~~~+\frac{c_1c_5(k+c_5c_7)}{c_6(c_1+c_4)}\mathrm{e}^{-c_2x}z(0,t)\nonumber\\
&~~~-\int_{0}^{c_4(s-\frac{x}{c_1})}\frac{c_1c_2c_5(k+c_5c_7)}{c_6(c_1+c_4)}\mathrm{e}^{-c_2(x+y)}z(y,t)dy\nonumber\\
&~~~-\int_{0}^{c_4(s-\frac{x}{c_1})}\frac{c_1c_2(k+c_5c_7)}{c_1+c_4}\mathrm{e}^{-c_2x}\psi(y,0,t)dy\nonumber\\
&~~~-\frac{kc_1c_5}{c_6(c_1+c_4)}\mathrm{e}^{-c_2(x+c_4s)}z(x+c_4s,t)\nonumber\\
&~~~+\frac{kc_1c_5}{c_6(c_1+c_4)}\mathrm{e}^{\frac{c_2(c_4-c_1)}{c_1}x-c_2c_4s}z(c_4(s-\frac{x}{c_1}),t)\nonumber\\     
&~~~-\int_{c_4(s-\frac{x}{c_1})}^{x+c_4s}\frac{kc_1c_2c_4c_5}{c_6(c_1+c_4)^2}\mathrm{e}^{-\frac{c_1c_2}{c_1+c_4}x-\frac{c_2c_4}{c_1+c_4}(y+c_1s)}\nonumber\\
&~~~\times z(y,t)dy\nonumber\\  
&~~~-\int_{c_4(s-\frac{x}{c_1})}^{x+c_4s} \frac{kc_1c_2}{c_1+c_4}\mathrm{e}^{-\frac{c_1c_2}{c_1+c_4}(x-y+c_4s)}\nonumber\\
&~~~\times\psi(y,0,t)dy\nonumber\\
&~~~-\frac{c_4c_5c_7}{c_6}\mathrm{e}^{-c_2x}v_x(c_{4}(s-\frac{x}{c_1}) ,t)\nonumber\\
&~~~+\frac{c_5^2c_7}{c_6}\mathrm{e}^{-c_2(c_4s+\frac{c_1-c_4}{c_1}x)}z(c_4(s-\frac{x}{c_1}),t)\nonumber\\
&~~~+c_5c_7\mathrm{e}^{-c_2x}\psi(c_{4}(s-\frac{x}{c_1}),0,t)\nonumber\\
&~~~-\frac{kc_4}{c_6}v_x(x+c_4s,t)\nonumber\\
&~~~+\frac{kc_5}{c_6}\mathrm{e}^{-c_2(x+c_4s)}z(x+c_4s,t)\nonumber\\
&~~~+ k\psi(x+c_4s,0,t)\nonumber\\
&~~~-\frac{c_1c_5c_7(k+c_5c_7)}{c_6(c_1+c_4)}\mathrm{e}^{-c_2x} v(c_4(s-\frac{x}{c_1}),t)\nonumber\\
&~~~+\frac{c_1c_5c_7(k+c_5c_7)}{c_6(c_1+c_4)}\mathrm{e}^{-c_2x} v(0,t)\nonumber\\
&~~~+\int_{0}^{c_4(s-\frac{x}{c_1})}\frac{c_1c_2c_5(k+c_5c_7)}{c_6(c_1+c_4)}\mathrm{e}^{-c_2(x+y)}z(y,t)dy\nonumber\\
&~~~+\int_{0}^{c_4(s-\frac{x}{c_1})}\frac{c_1c_5c_7(k+c_5c_7)}{c_4(c_1+c_4)}\mathrm{e}^{-c_2x}\psi(y,0,t)dy\nonumber\\
&~~~-\int_{0}^{\frac{x}{c_1}}c_1c_2\mathrm{e}^{-c_1c_2r}\psi_t(x-c_{1}r,s-r,t)dr\nonumber\\
&~~~+\int_{\frac{x}{c_1}}^sc_5c_7\mathrm{e}^{-c_2x}\psi_t(c_{4}(r-\frac{x}{c_1}),s-r,t)dr\nonumber\\
&~~~+\int_{0}^s k\psi_t(x+c_4r,s-r,t)dr\nonumber\\  
&~~~-\int_{0}^s \int_{\max\{x-c_1r,c_{4}(r-\frac{x}{c_1})\}}^{x+c_{4}r} \frac{kc_1c_2}{c_1+c_4}\mathrm{e}^{-\frac{c_1c_2}{c_1+c_4}(x-y+c_4r)}\nonumber\\
&~~~\times\psi_t(y,s-r,t)dydr;
\end{align}
in the case of  $c_1s < x \leq L-c_4s $,
\begin{align}
&\beta_t(x,s,t)=\psi_t(x,s,t)\nonumber\\
&~~~-\frac{c_1c_5}{c_6}\mathrm{e}^{-c_2x}z_x(x-c_1s,t)\nonumber\\
&~~~-c_1c_2\mathrm{e}^{-c_1c_2s}\psi(x-c_1s,0,t)\nonumber\\
&~~~-\frac{kc_1c_5}{c_6(c_1+c_4)}\mathrm{e}^{-c_2(x+c_4s)}z(x+c_4s,t)\nonumber\\
&~~~+\frac{kc_1c_5}{c_6(c_1+c_4)}\mathrm{e}^{-c_2x}z(x-c_1s,t)\nonumber\\
&~~~-\int_{x-c_1s}^{x+c_4s}\frac{kc_1c_2c_4c_5}{c_6(c_1+c_4)^2}\mathrm{e}^{-\frac{c_1c_2}{c_1+c_4}x-\frac{c_2c_4}{c_1+c_4}(y+c_1s)}z(y,t)dy\nonumber\\
&~~~-\int_{x-c_1s}^{x+c_4s}\frac{kc_1c_2}{c_1+c_4}\mathrm{e}^{-\frac{c_1c_2}{c_1+c_4}(x-y+c_4s)} \nonumber\\
&~~~\times\psi(y,0,t)dy\nonumber\\
&~~~-\frac{kc_4}{c_6}v_x(x+c_4s,t)\nonumber\\
&~~~+\frac{kc_5}{c_6}\mathrm{e}^{-c_2(x+c_4s)}z(x+c_4s,t)\nonumber\\
&~~~+k \psi(x+c_4s,0,t)\nonumber\\
&~~~-\int_0^s c_1c_2\mathrm{e}^{-c_1c_2r} \psi_t(x-c_1r,s-r,t)dr\nonumber\\
&~~~+\int_0^s k \psi_t(x+c_4r,s-r,t)dr\nonumber\\
&~~~-\int_0^s \int_{x-c_1r}^{x+c_4r}\frac{kc_1c_2}{c_1+c_4}\mathrm{e}^{-\frac{c_1c_2}{c_1+c_4}(x-y+c_4r)} \nonumber\\
&~~~\times\psi_t(y,s-r,t)dydr;
\end{align}
in the case of  $L-c_4s < x \leq L$,
\begin{align} &\beta_t(x,s,t)=\psi_t(x,s,t)\nonumber\\
&~~~-\frac{c_1c_5}{c_6}\mathrm{e}^{-c_2x}z_x(x-c_1s,t)\nonumber\\
&~~~-c_1c_2\mathrm{e}^{-c_1c_2s} \psi(x-c_1s,0,t)\nonumber\\
&~~~-\frac{kc_1c_5}{c_6(c_1+c_4)}\mathrm{e}^{-c_2L}z(\frac{c_1+c_4}{c_4}L-\frac{c_1}{c_4}x-c_1s,t)\nonumber\\
&~~~+\frac{kc_1c_5}{c_6(c_1+c_4)}\mathrm{e}^{-c_2x}z(x-c_1s,t)\nonumber\\
&~~~-\int_{x-c_1s}^{\frac{c_1+c_4}{c_4}L-\frac{c_1}{c_4}x-c_1s}\frac{kc_1c_2c_4c_5}{c_6(c_1+c_4)^2}\nonumber\\
&~~~\times\mathrm{e}^{-\frac{c_1c_2}{c_1+c_4}x-\frac{c_2c_4}{c_1+c_4}(y+c_1s)}z(y,t)dy\nonumber\\
&~~~-\int_{x-c_1s}^{\frac{c_1+c_4}{c_4}L-\frac{c_1}{c_4}x-c_1s}\frac{kc_1c_2}{c_1+c_4}\nonumber\\
&~~~\times\mathrm{e}^{-\frac{c_1c_2}{c_1+c_4}(x-y+c_4s)} \psi(y,0,t)dy\nonumber\\
&~~~-\frac{kc_5}{c_6}\mathrm{e}^{c_2L}z(L,t)\nonumber\\ 
&~~~+\frac{kc_5}{c_6}\mathrm{e}^{c_2L}z(\frac{c_1+c_4}{c_4}L-\frac{c_1}{c_4}x-c_1s,t)\nonumber\\ 
&~~~-\int_{\frac{c_1+c_4}{c_4}L-\frac{c_1}{c_4}x-c_1s}^{L} kc_2\mathrm{e}^{-c_2L} \psi(y,0,t)dy\nonumber\\ 
&~~~+\frac{kc_5}{c_6}\mathrm{e}^{c_2L}z(L,t)\nonumber\\  
&~~~+ k \psi(L,0,t)dr\nonumber\\   
&~~~-\int_0^s c_1c_2\mathrm{e}^{-c_1c_2r} \psi_t(x-c_1r,s-r,t)dr\nonumber\\
&~~~+\int_0^{\frac{L-x}{c_4}} k \psi_t(x+c_4r,s-r,t)dr\nonumber\\
&~~~+\int_{\frac{L-x}{c_4}}^s k \psi_t(L,s-r,t)dr\nonumber\\
&~~~-\int_0^s \int_{x-c_1r}^{\min\{x+c_4r,\frac{c_1+c_4}{c_4}L-\frac{c_1}{c_4}x-c_1r\}}\frac{kc_1c_2}{c_1+c_4}\nonumber\\
&~~~\times\mathrm{e}^{-\frac{c_1c_2}{c_1+c_4}(x-y+c_4r)} \psi_t(y,s-r,t)dydr\nonumber\\
&~~~-\int_{\frac{L-x}{c_4}}^s \int_{\frac{c_1+c_4}{c_4}L-\frac{c_1}{c_4}x-c_1r}^{L} kc_2\mathrm{e}^{-c_2L} \psi_t(y,s-r,t)dydr.
\end{align}
For each case,  one can reach  $\beta_t(x,s,t)-\beta_s(x,s,t)=0$.
In order to get    \eqref{equ-z1}, \eqref{equ-v1} and  \eqref{equ-bnd-v1} of the target system, we get the  relation between $\psi$ and $\beta$ at  $s=0$     via  transformation     \eqref{trans-beta} as follows: 
\begin{align}\label{beta0-psi0}
\psi(x,0,t)=\beta(x,0,t)-\frac{c_5}{c_6}
\mathrm{e}^{-c_2x}z(x,t)+\frac{k}{c_6}v(x,t).
\end{align} 
Substitute \eqref{beta0-psi0} into  
\eqref{equ-z}, \eqref{equ-v} and \eqref{equ-bnd-v}, respectively, which gives   \eqref{equ-z1}, \eqref{equ-v1} and  \eqref{equ-bnd-v1}. Hence, the transformation \eqref{trans-beta}   can transform the original system \eqref{equ-z}-\eqref{equ-bnd-psi} into the  target system (\ref{equ-z1})-(\ref{equ-bnd-beta}).

\section*{Acknowledgment}
The authors thank Dr. Nikolaos Bekiaris-Liberis for his valuable suggestions on the paper. 

\ifCLASSOPTIONcaptionsoff
\newpage
\fi

\bibliographystyle{plain}
\bibliography{bibfile}   

	\vspace{-2mm}
\begin{IEEEbiography}
[{\includegraphics
[width=1in,height=1.25in,clip,keepaspectratio]{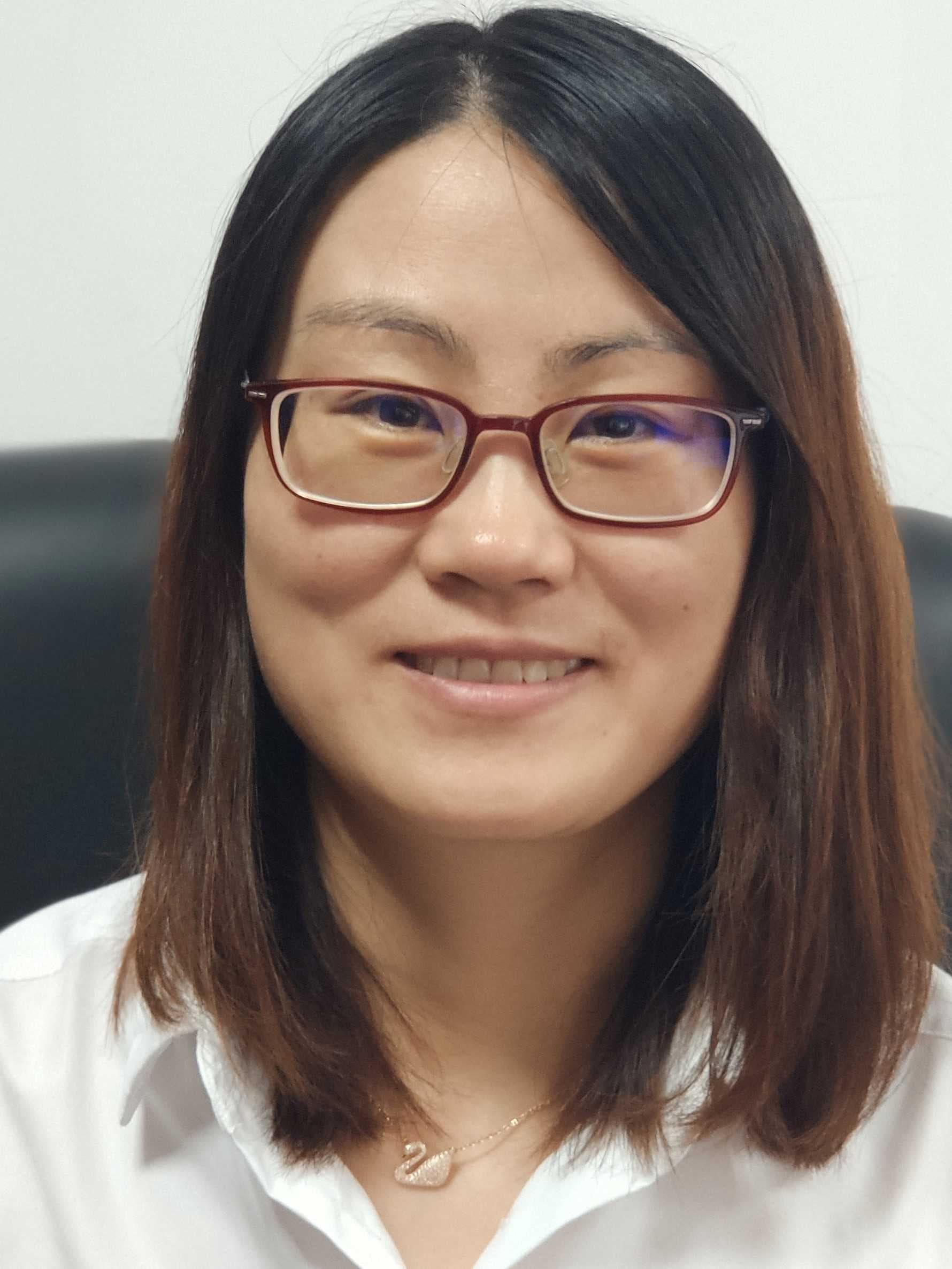}}]
{Jie Qi}
is Professor of Automation Department at the Donghua University, China. She received her Ph.D. degree in Systems Engineering (2005) and the B.S. degree in Automation (2000) from Northeastern University in Shenyang, China.   She  has been a research fellow with the Institute of Textiles \& Clothing, the Hong Kong Polytechnic University, Hong Kong from 2007 to 2008;  a visiting researcher with  the Cymer Center for Control Systems and Dynamics at the University of California, San Diego, from March 2013 to February 2014 and from June to September in 2015; and a visiting  researcher  with  the Chemical and Materials Engineering Department at the University of Alberta, from January  2019 to January  2020.  Her research interests include  control and estimation  of distributed parameters systems, control of delayed systems and its applications on  multi-agent systems and traffic systems.
	\vspace{-8mm}
\end{IEEEbiography}

\begin{IEEEbiography}
	[{\includegraphics
[width=1in,height=1.22in,clip,keepaspectratio]{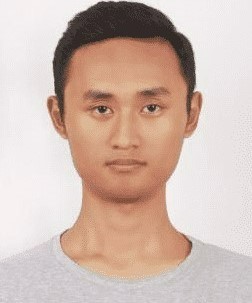}}]
{Shurong Mo}
received his B.S. degree  in Automation from Donghua University in 2018. He participated in the National Undergraduate Electronic Design Contest and won the third prize, in 2017. He is currently pursuing  his M.S. degree in Control Science and Control Engineering at Donghua University. His research interests include control of delayed systems, reinforcement learning-based control of traffic systems.  
	\vspace{-7mm}
\end{IEEEbiography}

\begin{IEEEbiography}[{\includegraphics
[width=1in,height=1.25in,clip,keepaspectratio]{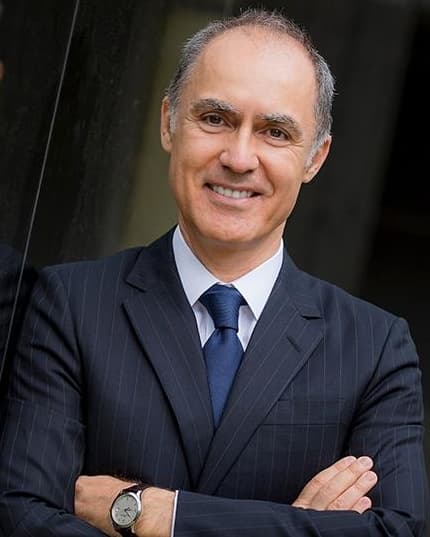}}]
{Miroslav Krstic}
is Distinguished Professor of Mechanical and Aerospace Engineering, holds the Alspach endowed chair, and is the founding director of the Cymer Center for Control Systems and Dynamics at UC San Diego. He also serves as Senior Associate Vice Chancellor for Research at UCSD. As a graduate student, Krstic won the UC Santa Barbara best dissertation award and student best paper awards at CDC and ACC. Krstic has been elected Fellow of seven scientific societies - IEEE, IFAC, ASME, SIAM, AAAS, IET (UK), and AIAA (Assoc. Fellow) - and as a foreign member of the Serbian Academy of Sciences and Arts and of the Academy of Engineering of Serbia. He has received the Richard E. Bellman Control Heritage Award, SIAM Reid Prize, ASME Oldenburger Medal, Nyquist Lecture Prize, Paynter Outstanding Investigator Award, Ragazzini Education Award, IFAC Nonlinear Control Systems Award, Chestnut textbook prize, Control Systems Society Distinguished Member Award, the PECASE, NSF Career, and ONR Young Investigator awards, the Schuck (’96 and ’19) and Axelby paper prizes, and the first UCSD Research Award given to an engineer. Krstic has also been awarded the Springer Visiting Professorship at UC Berkeley, the Distinguished Visiting Fellowship of the Royal Academy of Engineering, the Invitation Fellowship of the Japan Society for the Promotion of Science, and four honorary professorships outside of the United States. He serves as Editor-in-Chief of Systems \& Control Letters and has been serving as Senior Editor in Automatica and IEEE Transactions on Automatic Control, as editor of two Springer book series, and has served as Vice President for Technical Activities of the IEEE Control Systems Society and as chair of the IEEE CSS Fellow Committee. Krstic has coauthored fifteen books on adaptive, nonlinear, and stochastic control, extremum seeking, control of PDE systems including turbulent flows, and control of delay systems.

\end{IEEEbiography}

\end{document}